\documentclass[runningheads]{llncs}
\usepackage{amssymb}
\usepackage{mathrsfs,amsmath}
\usepackage{amsmath}
\usepackage{bbm}
\usepackage{dsfont}
\usepackage{listing}
\usepackage{hyperref}
\usepackage{color}
\usepackage[ruled, vlined,linesnumbered, noend]{algorithm2e}
\usepackage[acronym]{glossaries}
\usepackage[htt]{hyphenat} 
\usepackage{subcaption} 

\newtheorem{observation}{Observation}
\newtheorem{assumption}{Assumption}
\newtheorem{fact}{Fact}

\newacronym{offline}{\texttt{$k$-Max-Matching}}{}
\newacronym{offline:min}{\texttt{$k$-Min-Matching}}{}
\newacronym{online:max}{\texttt{Online-$(k,d)$-Max-Matching}}{} 
\newacronym{online:min}{\texttt{Online-$(k,d)$-Min-Matching}}{}

\newcommand{\offlinemax}{\texttt{$k$-Max-Matching}}
\newcommand{\offlinemin}{\texttt{$k$-Min-Matching}}
\newcommand{\onlinemax}{\texttt{Online-$(k,d)$-Max-Matching}}
\newcommand{\onlinemin}{\texttt{Online-$(k,d)$-Min-Matching}}

\newif \ifmargincomments 
\margincommentstrue 

\ifmargincomments

\else

\fi

\newif \ifappendix
\appendixtrue

\usepackage{graphicx}
\graphicspath{ } 

\newcommand{\abs}[1]{\left| #1 \right|}
\newcommand{\bigpar}[1]{\left( #1 \right)}
\newcommand{\bigbra}[1]{\left[ #1 \right]}
\newcommand{\bigbrace}[1]{\left\{ #1 \right\}}

\newcommand{\casewise}[1]{\left\{ #1 \right.}
 
\newcommand{\norm}[1]{\left| \left| #1 \right| \right|}

\newcommand{\cA}{\mathcal{A}}

\newcommand{\cE}{\mathcal{E}}

\newcommand{\cG}{\mathcal{G}}

\newcommand{\cM}{\mathcal{M}}

\newcommand{\cS}{\mathcal{S}}

\newcommand{\cU}{\mathcal{U}}

\begin{document}
\title{Online Hypergraph Matching with Delays
}
%
%
\author{
Marco Pavone\inst{1} \and
Amin Saberi\inst{1} \and 
Maximilian Schiffer\inst{2} \and 
Matthew Tsao\inst{1}} 
\authorrunning{M. Pavone et al.}
%
\institute{Stanford University, Stanford CA, 94305 USA \\ \texttt{\{pavone,saberi,mwtsao\}@stanford.edu} \and
Technical University of Munich, Munich 80333, Germany \\ \texttt{schiffer@tum.de}} 

\maketitle              
\begin{abstract}
We study an online hypergraph  matching problem with delays, motivated by ridesharing applications. In this model, users enter a marketplace sequentially, and are willing to wait up to $d$ timesteps to be matched, after which they will leave the system in favor of an outside option. A platform can match groups of up to $k$ users together, indicating that they will share a ride. Each group of users yields a match value depending on how compatible they are with one another. As an example, in ridesharing, $k$ is the capacity of the service vehicles, and $d$ is the amount of time a user is willing to wait for a driver to be matched to them. 

We present results for both the utility maximization and cost minimization variants of the problem. In the utility maximization setting, the optimal competitive ratio is $\frac{1}{d}$ whenever $k \geq 3$, and is achievable in polynomial-time for any fixed $k$. In the cost minimization variation, when $k = 2$, the optimal competitive ratio for deterministic algorithms is $\frac{3}{2}$ and is achieved by a polynomial-time thresholding algorithm. When  $k>2$, we show that a polynomial-time randomized batching algorithm is $(2 - \frac{1}{d}) \log k$-competitive, and  it is NP-hard to achieve a competitive ratio better than $\log k - O \bigpar{\log \log k}$.  
\end{abstract}
\section{Introduction}


The United States loses hundreds of billions of dollars annually through time spent waiting in traffic. Ridehailing services such as Uber and Lyft are significant contributors to metropolitan congestion, since the convenience and privacy offered by their services has led to increased travel demand \cite{Erhardt19}. 
At the same time, they also have the potential to reduce congestion by the way of \textit{shared rides}, whereby travelers can share rides instead of each having their own vehicle on the road.  A recent study \cite{Ostrovsky19} suggests that ridesharing can reduce congestion when deployed together with congestion pricing and road tolls. 


With this as motivation, we study the ridesharing problem through the lens of online hypergraph matching. In our  model, the problem instance is represented as a hypergraph. Vertices represent users who arrive to a marketplace over time and hyperedges joining multiple users signify that those users can participate in a shared service. To model disutility for waiting, users will wait in the system for at most $d$ timesteps, after which they will leave the system in favor of an outside option.

This model captures several characteristics of modern ridesharing. First, hyperedges represent groups of potentially more than $2$ people sharing a ride, i.e. a service vehicle can carry up to $k$ customers at a time. This is a generalization of previous work \cite{AshlagiBDJSS19} where at most $2$ people were allowed to share a ride. Similar to recent works \cite{AshlagiBDJSS19,Alonso-MoraSWFR17}, the spatial nature of ridesharing is captured by weights of the hyperedges, where the weight of a hyperedge measures the match efficiency of the associated requests. The online aspect of our model captures the temporal uncertainty faced in ridehailing systems where decisions must be made without knowledge of future demand. 


We study both the utility maximization and cost minimization variants of the  problem. In the utility maximization setting, edge weights represent the utility gained by the associated users sharing a service, so the goal is to find a matching in the hypergraph with large total weight. In the cost minimization setting, edge weights correspond to the cost required to serve all associated vertices together. Vertices also have weights, representing the cost of serving the vertex individually. The goal in this setting is to find a matching which minimizes the total edge cost plus the total cost of unmatched vertices. 

\subsection{Contributions}\label{subsec:contribution}

We study \onlinemax{} and \onlinemin{}, which are the utility maximization and cost minimization variants of the problem respectively. 
For \onlinemax{} we show that whenever $k \geq 3$, the optimal competitive ratio is $\frac{1}{d}$, and we present a polynomial-time randomized batching algorithm which is $\frac{1}{d}$-competitive. 

The cost minimization version of the problem will help us incorporate additional structures. We will assume the costs are monotone and subadditive.  That will allow us to prove the following set of results. For $k=2$, the optimal competitive ratio is $\frac{3}{2}$ for deterministic algorithms. We present a polynomial-time deterministic thresholding algorithm which is $\frac{3}{2}$-competitive. For $k>2$, there exists a randomized batching algorithm which is $2 - \frac{1}{d}$ competitive. This algorithm, however, is not polynomial-time. Using a reduction from set cover, we show that it is NP-hard to have a competitive ratio better than $\log k - O(\log \log k)$. Leveraging the  reduction, we also construct a randomized greedy batching procedure and show that it is $(2 - \frac{1}{d}) \log k$-competitive, establishing the optimal polynomial-time competitive ratio up to a factor of $2 - \frac{1}{d}$.

\subsection{Organization}
The remainder of this paper is structured as follows. In Section \ref{sec:model} we formally define our online hypergraph matching model with deadlines. We discuss related work in Section \ref{sec:literature}. In Section \ref{sec:utilmax} we study the utility maximization variant of the model wherein we characterize the optimal competitive ratio. We study the cost minimization variant of the model in Section \ref{sec:costmin}. Section \ref{sec:conclusion} presents conclusions and directions for future work. 

\section{Model}\label{sec:model}

A total of $n$ requests will arrive to the platform sequentially, one per timestep. Each request is represented by a vertex, and the vertices are ordered corresponding to the order in which requests arrive in the system. A request must be matched within $d$ timesteps of its arrival, i.e. a request arriving at timestep $i$ must be matched by timestep $i+d-1$ at the latest. A vertex is \textbf{critical} if it will leave the system in the next timestep. Concretely, a vertex arriving at timestep $i$ will become critical at timestep $i+d-1$.  

The system can serve the requests in groups of size at most $k$, where $k$ is a constant independent of $n,d$. For example, in ridesharing, $k$ would be the capacity of a single service vehicle. The edge set $E$ represents compatible matches. Namely, $e \in E$ implies that the set $e$ of users can be served together.

%
The problem instance can thus be represented by a hypergraph $H = (V,E,w)$. The following definition is a generalization of shareability graphs from \cite{Alonso-MoraSWFR17,AshlagiBDJSS19}. 

\begin{definition}[Shareability Hypergraph]\label{def:shareability}
A hypergraph $H = (V,E,w)$ is a Shareability Hypergraph with parameters $(n,d,k)$ if and only if 1) The set of vertices is ordered $V = [n]$, and vertex $i$ arrives at timestep $i$, 2) $\text{diam}(e) < d$ for all $e \in E$, where $\text{diam}(e) := \max_{i,j \in e} \abs{i-j}$ 3) The rank of $H$ is at most $k$, i.e. $\abs{e} \leq k \text{ for all } e \in E$ and 4) $w : E \rightarrow \mathbb{R}_+$, i.e. all hyperedge weights $w(e)$ are non-negative.
\end{definition}

In the utility maximization setting, the weight $w(e)$ of an edge
$e$ is the utility gained by having the vertices in $e$ be served together. An optimal sharing strategy in a utility maximization problem corresponds to a maximum weight matching in $H$. In the cost minimization setting, the weight $w(e)$ of an edge $e$ represents the cost incurred when the vertices in $e$ are served together. An optimal sharing strategy in a cost minimization problem thus corresponds to a minimum weight matching in $H$, where the system also pays a cost for each unmatched vertex.
In this model we  assume that all edge weights are non-negative. 



 We now formally define the utility maximization version of online hypergraph matching with deadlines, which we will call the \onlinemax{} problem. 

\begin{definition}[\onlinemax{}]\label{def:online_util_max} The vertices of a shareability hypergraph $H = (V,E,w)$ with parameters $(n,d,k)$ are revealed sequentially, one per timestep. The vertex that arrives at timestep $t$ will disappear at timestep $t+d$. A hyperedge (along with its weight) $e \in E$ is revealed once all of its vertices have arrived. The \onlinemax{} problem is to maximize over matchings of $H$ the weight of the matching, subject to the constraint that a hyperedge can only be included after it is revealed but before any of its vertices disappear, and hyperedge inclusion is irrevocable. 
\end{definition}

The offline version of the problem is as follows: 

\begin{definition}[\offlinemax{}]\label{def:offline_util_max}
We use \offlinemax{} to refer to the offline version of \onlinemax{} where $d>n$ so the whole hypergraph can be observed before computing a matching.
\end{definition}


The cost minimization version, which we call the \onlinemin{} problem can be defined similarly. In this setting, to ensure there is a feasible solution, we assume each customer can be served individually, i.e. the shareability hypergraph $H$ satisfies $\bigbrace{i} \in E$ for every $i \in V$. In other words, the platform must decide by time $i+d-1$ whether to serve request $i$ individually or have it share a service with other vertices.

\begin{definition}[\onlinemin{}]\label{def:online_cost_min}
The vertices of a shareability hypergraph $H = (V,E,w)$ with parameters $(n,d,k)$ are revealed sequentially, one per timestep. Each vertex has a weight, which represents the cost of serving this vertex individually. The vertex that arrives at timestep $t$ will disappear at timestep $t+d$. A hyperedge (along with its weight) $e \in E$ is revealed once all of its vertices have arrived. The \onlinemin{} problem is to minimize over matchings of $H$ the weight of the matching plus the weight of all unmatched vertices, subject to the constraint that a hyperedge can only be included after it is revealed but before any of its vertices disappear, and hyperedge inclusion is irrevocable. 
\end{definition}

\begin{definition}[\offlinemin{}]\label{def:offline_cost_min}
We use \offlinemin{} to refer to the offline version of \onlinemin{} where $d>n$ so the whole hypergraph can be observed before computing a matching.
\end{definition}

\section{Related Work}\label{sec:literature}
This paper is related to online resource allocation and online matching problems with delays. In the following, we survey some recent results in these fields and discuss how they relate to this work. 

Online resource allocation has applications in advertising \cite{KorulaP09}, network routing problems \cite{Ma2018ACA}, and ridesharing \cite{AshlagiBDJSS19}. In general, obtaining polynomial-time constant factor approximations for these problems is NP-hard, so \cite{KorulaP09,Ma2018ACA} study special cases where users can request for at most $k$ different resource types. This assumption is realistic in most settings as a single user will not need more than a constant number of resources. \cite{KorulaP09} presents a sample-and-price algorithm and shows that it is $\frac{1}{k^2}$-competitive if the arrival order of requests is uniformly random. If the full distribution over user requests is known to the network operator, \cite{Ma2018ACA} show via approximate dynamic programming that a competitive ratio of $\frac{1}{k+1}$ is possible.


The above line of work does not model the trade-off between waiting costs and market thickness. The option of waiting for a thicker market is studied in \cite{AshlagiACCGKMWW17,EmekKW16,AkbarpourEtAl20,Aouad19,AshlagiBDJSS19} in different contexts. In these works, requests arrive sequentially in time, and the platform can match pairs of requests together. The objective in \cite{AshlagiACCGKMWW17,EmekKW16,AkbarpourEtAl20,Aouad19} is cost minimzation, and the objective in \cite{AshlagiBDJSS19} is utility maximization. 

In \cite{AshlagiACCGKMWW17,EmekKW16}, the platform can defer matching decisions but must pay a cost proportional to the total time requests wait before being matched. In \cite{AkbarpourEtAl20,Aouad19} requests arrive to the system according to some known distribution. Requests will stay in the system for a random and unkown amount of time, so the platform operator risks users leaving the system if it chooses to wait for a thicker market. Online matching with deadlines was studied in \cite{AshlagiBDJSS19}. Requests enter the system sequentially and are willing to wait up to $d$ timesteps to be matched, after which they will leave the system. Here $d$ is known to the platform. The authors present a $\frac{1}{4}$-competitive algorithm and also show that no algorithm can have a competitive ratio better than $\frac{1}{2}$. 

This paper contributes generalizations to these aforementioned works in two ways. With respect to the online matching literature, we consider a model where $k>2$ users can be matched together. This is motivated by ridesharing applications where service vehicles can hold more than $2$ customers at any given time. With respect to the online resource allocation literature, we present results that do not assume any distributional information on the input instances. Our model also allows for waiting so that the platform has $d$ units of time after a request appears to determine whether to accept it or reject it. This is in contrast to the works of \cite{KorulaP09,Ma2018ACA} where the platform must immediately decide between accepting or rejecting upon a request's arrival.

\section{Utility Maximization: \onlinemax{}}\label{sec:utilmax}

In this section, we characterize the optimal competitive ratio for \onlinemax{} by first proving an upper bound on achievable competitive ratios, and then presenting a polynomial-time algorithm which achieves this upper bound. 

\subsection{Upper Bounds}\label{sec:maxcost:algs:online:UB}
In this section we show that if the matching capacity $k$ is at least $3$, then no online algorithm can have a better competitive ratio than $\frac{1}{d}$, with $d$ being the number of timesteps a customer stays in the system. 

\begin{theorem}[Upper bound on the competitive ratio]\label{thm:online_UB}
If $k \geq 3$, then no online algorithm can have a competitive ratio better than $\frac{1}{d}$ for \onlinemax{}. 
\end{theorem}

\begin{proof}
Consider a family $\cG$ of shareability hypergraphs where $G \in \cG$ is of the form $G = (V,E,w)$ with $V = \bigbrace{0,1,...,2d-2}$, $w \in \mathbb{R}^{\abs{E}}_+$, and $E = \bigbrace{e_t}_{t=0}^{d-1}$ with $e_0 := \bigbrace{0,d-2,d-1}, e_{d-1} := \bigbrace{d-1,d,2d-2}$ and 
\begin{align*}
e_{t} := \bigbrace{t,d-1,t+d-1} \text{ for all } 0 < t < d-1.
\end{align*}

First, note that any matching can have at most one hyperedge since $e_t \cap e_{t'} \neq \emptyset$ for any $t,t'$ since in particular, all edges contain the vertex $d-1$. Second, every edge $e_t$ can only be chosen in the timestep in which its weight is revealed. This is because the earliest vertex in $e_t$ is $t$ so $e_t$  expires at time $t+d$ and the latest vertex in $e_t$ is $t+d-1$, so the earliest time $e_t$ can be chosen is $t+d-1$. So \onlinemax{} restricted to instances in $\cG$ is equivalent to the Adversarial Secretary Problem with $d$ applicants, which is defined below. See Appendix~\ref{dASP_vis} for a visualization of graphs in $\cG$ and the relationship to $d$-\texttt{ASP}. 
\begin{definition}[Adversarial Secretary Problem with $d$ applicants (\texttt{$d$-ASP})]\label{def:ASP}
Adversarial Secretary Problem with $d$ applicants, which we denote by \texttt{$d$-ASP}, has the following setting: There are $d$ applicants for a secretary position, with corresponding non-negative aptitude scores $\bigbrace{w_i}_{i=1}^d$.  At most $1$ applicant can be hired and decisions are irrevocable. Time is discretized into $d$ timesteps. The aptitude $w_i$ of the $i$th applicant is revealed at time $i$, and the $i$th applicant can only be hired at time $i$. The goal is to maximize the aptitude of the hired secretary. 
\end{definition}
Contrary to the canonical secretary problem where the arrival order of candidates is uniformly random, \texttt{$d$-ASP} makes no distributional assumption on the ordering of the candidates. Next we establish a hardness result for $d$-\texttt{ASP}. 
\begin{lemma}[Adversarial Secretary Problem]\label{lem:advsecretary}
The best possible competitive ratio for the adversarial secretary problem with $d$ applicants is $\frac{1}{d}$. 
\end{lemma}
\noindent See Appendix~\ref{pf:lem:advsecretary} for a proof of Lemma \ref{lem:advsecretary}. Since $d$-\texttt{ASP} is a special case of \onlinemax{}, this proves that the best possible competitive ratio of the latter is at most $\frac{1}{d}$. \qed 
\end{proof}

\subsection{Algorithms}\label{sec:maxcost:algs}

In this section we present a polynomial-time $\frac{1}{d}$-competitive algorithm for \onlinemax{} by adapting techniques from \offlinemax{}, the offline version of the problem. Specifically, we analyze the \texttt{Randomized-Batching} algorithm described in Algorithm \ref{alg:randombatch} which runs a \offlinemax{} algorithm $\cA$ every $d$ timesteps, effectively treating the online problem as $n/d$ separate offline problems. We refer to sets of $d$ consecutive vertices on which $\cA$ is applied as batches. Figure~\ref{fig:batchmatch} illustrates a batching procedure. The start and end times for the batches are determined by a random shift $Z \in \bigbrace{0,1,...,d-1}$.

The random shift $Z$ in \texttt{Randomized-Batching} is essential to the competitiveness of batching. Figure~\ref{fig:batchmatch2} illustrates why batching without randomization can do arbitrarily worse than the optimal, i.e. it has a competitive ratio of zero.

\begin{algorithm}
\caption{\texttt{Randomized-Batching}$(H;\cA)$}\label{alg:randombatch}
\textbf{Input:} Shareability Graph $H = (V,E,w)$, Offline matching algorithm $\cA$\;
\textbf{Output:} Matching $M$\;
$M \leftarrow \emptyset$\;
Draw $Z$ uniformly at random from $\bigbrace{0,1,2,...,d-1}$\;
\For{$1 \leq i \leq n$}{
Vertex $i$ is revealed\;
\If{$(i+1 \equiv Z \mod d) \text{ or } (i = n)$}{
Use $\cA$ to compute a matching $M_0$ for nodes appearing at times in $\bigbrace{i-d+t}_{t=1}^d$\;
$M \leftarrow M \cup M_0$\;
}	
}
\textbf{Return} $M$
\end{algorithm}

\begin{figure}[h]
	\centerline{
		\includegraphics[width=0.7\textwidth]{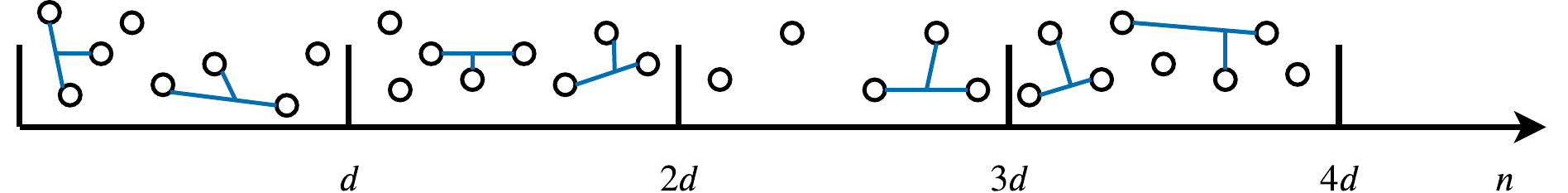}
	}
	\caption{An example of batching for the online matching problem. The time horizon is partitioned into batches of length $d$, and an offline algorithm is used to compute a matching (blue hyperedges) for each batch.}\label{fig:batchmatch}
\end{figure}

\begin{figure}[h]
	\centerline{
		\includegraphics[width=0.7\textwidth]{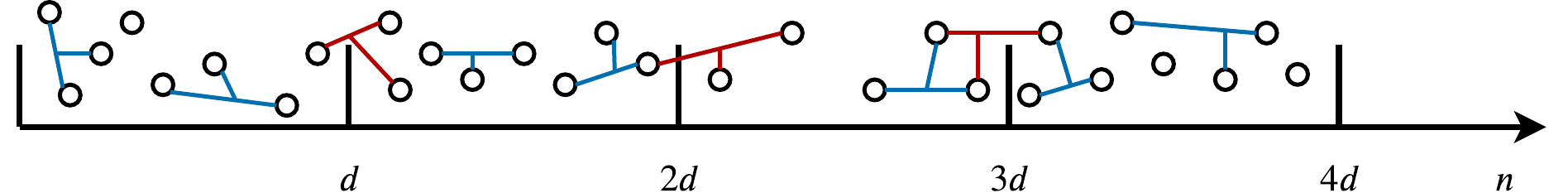}	
	}
	
	\caption{Consider a deterministic batching algorithm where the $i$th batch consists of vertices $\bigbrace{id + t}_{t=0}^{d-1}$. The batching algorithm is unable to select the red edges, since they are not fully contained in any batch. If the red edges each have weight $1$ and all other edges have weight $\epsilon$, then  the weight of the matching obtained by batching is at most $\epsilon \abs{E}$, whereas the optimal solution has weight at least $3$. Hence the competitive ratio of batching is at most $\frac{\epsilon \abs{E}}{3}$, which converges to zero as $\epsilon \rightarrow 0$.}\label{fig:batchmatch2}
	
\end{figure}

The following result provides a performance guarantee for \texttt{Randomized-Batching}.

\begin{theorem}[Performance of \texttt{Randomized-Batching}]\label{thm:raw_LB}
Let $\cA$ be an algorithm for \offlinemax{} with competitive ratio $\rho(\cA)$. For any matching $\widetilde{M}$ in $H$, the matching $M$ returned by Algorithm \ref{alg:randombatch} using $\cA$ satisfies
\begin{align*}
\mathbb{E} \bigbra{ w\bigpar{M} } \geq \frac{\rho \bigpar{ \cA }}{d} \sum_{e \in \widetilde{M}} w(e) \bigpar{ d - \text{diam}(e) }.
\end{align*}
\end{theorem}

\begin{proof} 
Let $\cA$ be an offline matching algorithm for \offlinemax{} with competitive ratio $\rho(\cA)$ and consider an instance $H = (V,E,w)$ of \onlinemax{}. For each $z \in \bigbrace{0,1,...,d-1}$ we define the batch $W_{z,i} = (V_{z,i}, E_{z,i},w)$ where
\begin{enumerate}
    \item $V_{z,i} := \bigbrace{z + id + t}_{t=0}^{d-1}$, i.e. all vertices arriving between timestesp $z+id$ and $z+(i+1)d-1$.
    \item $E_{z,i} := \bigbrace{e \in E : e \subset V_{z,i}}$, i.e. all edges whose vertices are all in $V_{z,i}$. 
\end{enumerate}
We will also define $E_z := \cup_{i} E_{z,i}$. With this notation, \texttt{Randomized-Batching} returns a matching $M := \cup_{i} M_{z,i}$, where $M_{z,i}$ is obtained by applying $\cA$ to $W_{z,i}$. We will also use $M_{z,i}^*$ to denote a maximum weight matching in $W_{z,i}$. We will make use of the following observation. 

\begin{observation}\label{obs:indicatorsum}
For any $e \in E$ we have $\sum_{z=0}^{d-1} \mathds{1}_{[e \in E_z]} \geq d - \text{diam}(e)$. 
\end{observation}

\begin{proof} 
Fix any $e \in E$ and let $t$ be the smallest arrival time of any vertex in $e$. Without loss of generality (by a linear change of coordinates) we assume that $t = 0$. By definition of $\text{diam}(e)$, all vertices of $e$ appear between time $0$ and $\text{diam}(e)$. Any edge that is a subset of $\bigbrace{z,z+1,...,z+d-1}$ will belong to $E_z$. Next, note that for any $z$ satisfying $\text{diam}(e)-d+1 \leq z \leq 0$, we have $z \leq 0$ and $z+d-1 \geq \text{diam}(e)$, and hence $e \subset \bigbrace{z,z+1,...,z+d-1}$. Therefore, $e \in E_z$ for every $z \in \bigbrace{\text{diam}(e)-d+1,...,-1,0}\mod d$. In particular, $e \in E_z$ for at least $d-\text{diam}(e)$ values of $z$. \qed 
\end{proof}

\noindent Let $\widetilde{M}$ be a matching for $H$, and let $M_z^*$ be an optimal matching for $H_z := (V,E_z,w)$. Given $Z = z$, we have
\begin{align}\label{eqn:rand_batch_Z=z}
    w(M) = \sum_{i} w(M_{z,i}) &\overset{(a)}{\geq} \sum_{i} \rho(\cA) w(M_{z,i}^*) = \rho(\cA) w(M_{z}^*). 
\end{align}

Where $(a)$ is because $\cA$ is $\rho(\cA)$-competitive, and $M_{z,i}$ is obtained by applying $\cA$ to $W_{z,i}$. Because $Z$ is uniformly distributed over $\bigbrace{0,1,...,d-1}$, we have
\begin{align*}
\mathbb{E}_Z \bigbra{ w(M) } &= \sum_{z=0}^{d-1} \mathbb{P} \bigbra{ Z = z } \mathbb{E} \bigbra{ w(M) | Z = z} \overset{(a)}{\geq} \frac{1}{d} \sum_{z=0}^{d-1} \rho(\cA) \sum_{e^* \in M^*_z} w(e^*) \overset{(b)}{\geq} \frac{\rho(\cA)}{d} \sum_{z=0}^{d-1} \sum_{e^* \in \widetilde{M} \cap E_z} w(e^*) \\
&= \frac{\rho(\cA)}{d} \sum_{z=0}^{d-1} \sum_{e^* \in \widetilde{M}} w(e^*) \mathds{1}_{[e \in E_z]} = \frac{\rho(\cA)}{d} \sum_{e^* \in \widetilde{M}} w(e^*) \sum_{z=0}^{d-1} \mathds{1}_{[e \in E_z]} \overset{(c)}{\geq} \frac{\rho(\cA)}{d}  \sum_{e^* \in \widetilde{M}} w(e^*) \bigpar{d - \text{diam}(e)}.
\end{align*}
Where $(a)$ is due to \eqref{eqn:rand_batch_Z=z}, $(b)$ is because and $\widetilde{M} \cap E_z$ is a matching in $H_z$ and $M^*_z$ is a maximum weight matching in $H_z$. Finally, $(c)$ is due to Observation \ref{obs:indicatorsum}.  \qed 
\end{proof}

\noindent The following is a straightforward consequence of Theorem \ref{thm:raw_LB} by choosing $\widetilde{M}$ to be an optimal matching, and noting that $d - \text{diam}(e) \geq 1$ for all $e \in E$ since $H$ is a shareability hypergraph. 

\begin{lemma}[Competitive Ratio of \texttt{Randomized-Batching}]\label{lem:online}
Let $\cA$ be an offline matching algorithm for \offlinemax{} with competitive ratio $\rho(\cA)$. Algorithm \ref{alg:randombatch} using $\cA$ is $\frac{\rho(\cA)}{d}$-competitive for \onlinemax{}.
\end{lemma} 

\noindent We now briefly explore the implications of Lemma \ref{lem:online} for different choices of $\cA$. The best known polynomial-time algorithms for \offlinemax{} are $O(\frac{1}{k})$-competitive \cite{Berman00,ChandraH01a}, and in particular \texttt{Greedy} is $\frac{1}{k}$-competitive.\footnote{See Appendix~\ref{sec:prelims:hypergraph} for further details and existing results for \offlinemax{}.}  

\begin{corollary}\label{cor:1/kd}
Choosing $\cA$ to be \texttt{Greedy} when running Algorithm \ref{alg:randombatch} gives a polynomial-time algorithm with competitive ratio of $\frac{1}{kd}$ for \onlinemax{}. 
\end{corollary}

\noindent We can also consider an exact procedure $\cA^*$ for \offlinemax{} which satisfies $\rho(\cA^*) = 1$. Applying Lemma \ref{lem:online} with such an $\cA^*$ establishes $\frac{1}{d}$ as the optimal competitive ratio for \onlinemax{}.  

\begin{corollary}\label{cor:1/d}
If $\cA^*$ is an exact procedure for \offlinemax{}, i.e., $\rho(\cA^*) = 1$, then Algorithm \ref{alg:randombatch} using $\cA^*$ as a subroutine achieves a competitive ratio of $\frac{1}{d}$ for \onlinemax{}. 
\end{corollary}
\noindent Corollary \ref{cor:1/d} in conjunction with Theorem~\ref{thm:online_UB} shows that the optimal competitive ratio of $\frac{1}{d}$ is achievable in finite time. However, it is known from \cite{Hazan06} that achieving $\rho(\cA) = \Omega \bigpar{ \frac{\log k}{k} }$ for \offlinemax{} is NP-hard, and thus no exact procedure $\cA^*$ can run in polynomial-time unless $\text{P = NP}$. For this reason, Corollary \ref{cor:1/d} does not show that a competitive ratio of $\frac{1}{d}$ is achievable in \textit{polynomial-time}. 

To obtain a polynomial-time $\frac{1}{d}$-competitive algorithm, we will replace \texttt{Greedy} by \texttt{Depth-$k$ Greedy} as the choice for $\cA$. \texttt{Depth-$k$ Greedy} is described in Algorithm \ref{alg:dk_greedy}. Making this replacement leads to the following result. 

\begin{theorem}\label{thm:online_polyLB}
\texttt{Randomized-Batching} using $\cA =$\texttt{Depth-$k$ Greedy} as a subroutine is $\frac{1}{d}$-competitive for \onlinemax{}. Furthermore, its running time is $O \bigpar{ \abs{E}^{k+1} }$ which is polynomial in the size of the shareability hypergraph $H$ for any fixed $k$. 
\end{theorem}

\begin{algorithm}
\caption{\texttt{Depth-$k$ Greedy}$(H)$}\label{alg:dk_greedy}
\textbf{Input:} Hypergraph $H = (V,E,w)$ where the vertices are ordered\;
\textbf{Output:} A matching $M$ for $H$ \;
Define $L$ to be the first $k$ vertices\;
Define $R$ to be the last $k$ vertices\;
Define $\cM_{L,R}(H) := \bigbrace{M \text{ is a matching}: e \cap L \neq \emptyset, e \cap R \neq \emptyset \; \forall e \in M}$\;
Enumerate all matchings in $\cM_{L,R}(H)$ as $\bigbrace{M_i}_{i=1}^{\abs{\cM_{L,R}(H)}}$\;
\For{$1 \leq i \leq \abs{\cM_{L,R}(H)}$}{
$V_i \leftarrow \bigbrace{i \in V : i \text{ is matched in }M_i}$\;
$\widetilde{M}_i \leftarrow \hyperref[alg:offlinegreedy]{\texttt{Greedy}}((V \setminus V_i, E))$ 
}
$i^* \leftarrow \arg\max_i w(M_i \cup \widetilde{M}_i)$\;
$M \leftarrow M_{i^*} \cup \widetilde{M}_{i^*}$ \;
\textbf{Return} $M$
\end{algorithm}
In the remainder of this section, we sketch the proof of Theorem~\ref{thm:online_polyLB}. See Appendix~\ref{pf:thm:online_polyLB} for a full proof of Theorem~\ref{thm:online_polyLB}.

To understand how to construct a polynomial-time algorithm for \onlinemax{} with the optimal competitive ratio of $\frac{1}{d}$, we first need to understand why Corollary \ref{cor:1/kd} falls short: let $M^*$ be a maximum weight matching for a shareability hypergraph $H = (V,E,w)$. We can partition $M^*$ into two sets: \textit{short hyperedges} $M_s^*$ and \textit{long hyperedges} $M_\ell^*$ as follows:
\begin{align*}
M_{s}^* &:= \bigbrace{e \in M^* : \text{diam}(e) \leq d-k} \text{ and } M_{\ell}^* := \bigbrace{e \in M^* : \text{diam}(e) > d-k}.
\end{align*}
As the names suggest, $M_s^*$ contains all hyperedges of $M^*$ with diameter at most $d-k$, and $M_\ell^*$ contains all hyperedges whose diameters are larger than $d-k$. By Theorem~\ref{thm:raw_LB}, running \texttt{Randomized-Batching} with the \texttt{Greedy} subroutine results in a matching whose weight satisfies
\begin{align*}
    \mathbb{E} \bigbra{ w\bigpar{M} } &\geq \frac{\rho \bigpar{ \cA }}{d} \sum_{e \in M^*} w(e) \bigpar{ d - \text{diam}(e) } = \frac{1}{kd} \sum_{e \in M_s^*} w(e) \bigpar{ d - \text{diam}(e) } + \frac{1}{kd} \sum_{e \in M_\ell^*} w(e) \bigpar{ d - \text{diam}(e) } \\
    &\geq \frac{1}{kd} \sum_{e \in M_s^*} w(e) k + \frac{1}{kd} \sum_{e \in M_\ell^*} w(e) = \frac{1}{d} w(M_s^*) + \frac{1}{kd} w(M_\ell^*).
\end{align*}

\noindent From this, we see that the approximation factor for $M_s^*$ and $M_\ell^*$ are different. 

\begin{observation}\label{obs:weaklink}
Since $M_s^*$ is already contributing a $\frac{1}{d}$ fraction of its weight to $M$, we just need to improve the contribution of $M_\ell^*$ from $\frac{1}{kd}$ to $\frac{1}{d}$ in order to achieve a competitive ratio of $\frac{1}{d}$. 
\end{observation}

\noindent Since all edges in $M_\ell^*$ have large diameter, we can show that $\abs{M_\ell^*}$ cannot be too large. Specifically, for any batch (i.e. subgraph on $d$ consecutive vertices) $(V',E')$, we can show that $\abs{M_\ell^* \cap E'} \leq k$. We can then show that there are at most $O \bigpar{ {d \choose k}^k }$ possibilities for $M_\ell^* \cap E'$. Since $k$ is a constant, this is polynomial in the batch size. Finally, we extend all of these possibilities for $M_\ell^* \cap E'$ to maximal matchings in $(V',E')$ via \texttt{Greedy} and return the maximal matching with the highest total weight. This can be done in polynomial-time, and improves the contribution of $w(M_\ell^*)$ from $\frac{1}{kd}$ to $\frac{1}{d}$. This motivates \texttt{Depth-$k$ Greedy} described in Algorithm \ref{alg:dk_greedy} as an alternative to \texttt{Greedy}. 

\section{Cost Minimization: \onlinemin{}}\label{sec:costmin}

In this section we study \onlinemin{} specialized to the ridesharing problem. We first present additional assumptions that we make to utilize the structure of ridesharing problems. For $k=2$ we prove that the optimal competitive ratio for deterministic algorithms is $\frac{3}{2}$ and is achieved by a thresholding algorithm. For $k>2$, we characterize both the optimal competitive ratio and the optimal competitive ratio attainable with polynomial-time algorithms up to a factor of $2 - \frac{1}{d}$. 


\subsection{Ridesharing Assumptions}\label{sec:costmin:rh_assumption}

\noindent Let $H = (V,E,w)$ be the shareability hypergraph for an instance of \onlinemin{}. For the application to ridesharing, we leverage the following two assumptions.

\begin{assumption}[Monotonicity]\label{assumption:min:split}
For any $e \in E$, and any $e' \subseteq e$, we have $e' \in E$ and $w(e') \leq w(e)$. 
\end{assumption}


Assumption \ref{assumption:min:split} is realistic in ridesharing problems for the following reason. The cost $w(e)$ is the minimum cost required to service the customers in $e$, i.e. taking them from their pickup point to their dropoff point. Let $p$ be such a minimum cost path. Note that $p$ also serves all requests in $e'$, and therefore, the minimum cost required to serve $e'$ can only be smaller than $w(e)$. Thus we see that $w(e') \leq w(e)$. 


\begin{assumption}[Subadditivity]\label{assumption:min:sublinear}
For any $e \in E$, and any partition $e = \sqcup_{j=1}^\ell e_j$, we have $w(e) \leq \sum_{j=1}^\ell w(e_j)$. 
\end{assumption}

The above assumption is also quite reasonable. If there exists $e \in E$ which does not satisfy Assumption \ref{assumption:min:sublinear}, then any matching that includes $e$ can be converted to a matching with lower cost by replacing $e$ by $\bigbrace{e_j}_{j=1}^\ell$. Therefore $e$ will never be in an optimal matching, and so we can remove $e$ from $E$. Finally, we will consider the single occupancy matching $M_{so}$ as a benchmark. 

\begin{definition}[Single Occupancy Matching]
The single occupancy matching serves each customer individually, i.e. $M_{so} =\emptyset $. It is clear from Assumption~\ref{assumption:min:sublinear} that the competitive ratio of $M_{so}$ is $k$. 
\end{definition}
\subsection{The $k=2$ Case}\label{sec:costmin:k=2}

We first study the $k=2$ case where the shareability hypergraph is in fact a shareability graph. In this case, single occupancy is a $2$-competitive algorithm. We will thus limit our search to algorithms whose competitive ratio is smaller than $2$. We will use the following notation: 

\begin{itemize}
    \item For $i \in V$, we use $w_i \in \mathbb{R}_+$ to denote the cost of serving $i$ individually.
    \item For $(i,j) \in E$, we use $w(i,j) \in \mathbb{R}_+$ to denote the cost of serving $i$ and $j$ together. 
    \item Given a matching $M$, a vertex $i$ is unmatched under $M$ if it is not an endpoint of any edge in $M$.
    \item We use $w(M)$ to represent the cost of a matching $M$, defined via $w(M) := \sum_{(i,j) \in M} w(i,j) + \sum_{i \in U} w_i$ where $U$ is the set of unmatched vertices in $M$. 
\end{itemize}

The following lemma gives a lower bound on the competitive ratio of all deterministic and randomized algorithms and is proved in Appendix~\ref{pf:lem:mincost:detLB}.

\begin{lemma}[Lower Bound on the Competitive Ratio]\label{lem:mincost:detLB}
When $k=2$, no deterministic algorithm can have a competitive ratio smaller than $\frac{3}{2}$ and no randomized algorithm can have a competitive ratio smaller than $\frac{5}{4}$. 
\end{lemma}

Before presenting our algorithm, let us discuss a motivating example to highlight one of the key challenges in \onlinemin{}. Consider the following shareability graph where $w_A = 1$, $w_B = x > 1$, and $w(A,B) = x$. There is also a third vertex $C$ with $w_C \in \bigbrace{0,x}$ and $w(B,C) = x$. \\

\centerline{
\includegraphics[width=0.5\textwidth]{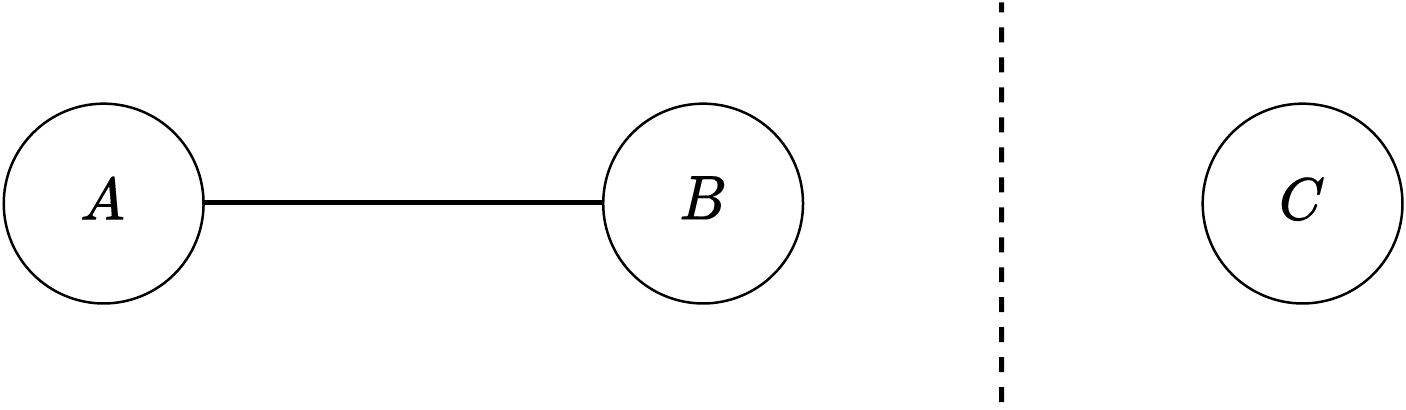}
}
Suppose $A$ becomes critical before $C$ arrives. In this case, our algorithm needs to decide whether to match $(A,B)$, or wait and match $(B,C)$ instead. The optimal matching depends on $w_C$. In particular, if $w_C = 0$, then the minimum cost is $w(A,B) + w_C = x$ and is achieved by matching $(A,B)$. However, if $w_C = x$, then the minimum cost is $w_A + w(B,C) = 1 + x$. Since we are interested in the worst-case error, we make the following two observations. 
\begin{enumerate}
    \item If $w_C = 0$, and we choose not to match $(A,B)$, then our resulting cost will be $w_A + w(B,C) = 1+x$. Recall that the minimum cost is $x$, so our matching has a cost that is $\frac{1+x}{x}$ times as large as the minimum cost. 
    \item If $w_C = x$ and we choose to match $(A,B)$, then our resulting cost will be $w(A,B) + w_C = 2x$. Recall that the minimum cost in this case is $1+x$, so our matching has a cost that is $\frac{2x}{1+x}$ times as large as the minimum cost.
\end{enumerate}

Since our goal is to minimize the worst case performance, this shows that we should match $(A,B)$ if $\frac{2x}{1+x} \leq \frac{1+x}{x}$, and not match $A$ otherwise. Equivalently, we should match $(A,B)$ if and only if $x \leq \frac{1}{\sqrt{2}-1}$. 


 This motivates thresholding algorithms which only accept matches that are sufficiently high quality as a way of balancing immediate reward with possible future reward. Algorithm \ref{alg:riskthreshold} describes the \texttt{Risk-Threshold} algorithm, which matches critical vertices to the most appealing partner according to a risk score. We also present \texttt{Risk-Threshold-ag} in Algorithm \ref{alg:riskthreshold_ag}, which is an agnostic version of \texttt{Risk-Threshold} that does not need to know when vertices become critical, and can thus be used in settings where users' willingness to wait is not observable. The competitive ratios of these algorithms are presented in the following results. 

\begin{algorithm}[h] 
    \caption{\texttt{Risk-Threshold-ag}$(H,\theta)$}\label{alg:riskthreshold_ag}
    \textbf{Input:} Shareability Graph $H = (V,E,w)$, Risk parameter $\theta \in [\frac{1}{2},1]$\;
    \textbf{Output:} Matching $M$\;
    Define $\theta_{ij} := \frac{w(i,j)}{w_i + w_j}$ for $(i,j) \in E$\;
    $M \leftarrow \emptyset$\;
    \While{There are vertices in the system}{
        \If{$\exists i,j : \theta_{ij} \leq \theta$}{
        $M \leftarrow M \cup \bigbrace{ (i,j) }$\;
        }
    }
    \textbf{Return} $M$
\end{algorithm}

\begin{algorithm}[h]
    \caption{\texttt{Risk-Threshold}$(H,\theta)$}\label{alg:riskthreshold}
    \textbf{Input:} Shareability Graph $H = (V,E,w)$, Risk parameter $\theta \in [\frac{1}{2},1]$\;
    \textbf{Output:} Matching $M$\;
    Define $\theta_{ij} := \frac{w(i,j)}{w_i + w_j}$ for $(i,j) \in E$\;
    $M \leftarrow \emptyset$\;
    \While{There are vertices in the system}{
    \If{$i$ is critical}{
    $j^* \leftarrow \arg\min_j \theta_{ij}$\;
    }
    \If{$\theta_{ij^*} \leq \theta$}{
    $M \leftarrow M \cup \bigbrace{(i,j^*)}$
    }
    }
    \textbf{Return} $M$
\end{algorithm}

\begin{theorem}
\label{thm:mincost:detUB_ag}
Using $\theta = \frac{\sqrt{5}-1}{2}$, the competitive ratio of \texttt{Risk-Threshold-ag}$(\cdot, \frac{\sqrt{5}-1}{2})$ is $\frac{\sqrt{5}+1}{2} \approx 1.61$. 
\end{theorem}
\begin{theorem}
\label{thm:mincost:detUB}
Using $\theta = \frac{2}{3}$, \texttt{Risk-Threshold}$(\cdot, \theta)$ is $\frac{3}{2}$-competitive. 
\end{theorem}

\noindent Theorem~\ref{thm:mincost:detUB} in conjunction with Lemma~\ref{lem:mincost:detLB} shows that \texttt{Risk-Threshold} has the optimal competitive ratio among deterministic algorithms. 

While every customer is willing to wait $d$ timesteps in our model, \texttt{Risk-Threshold-ag} and \texttt{Risk-Threshold} can handle heterogeneous deadlines. The algorithms do not require homogeneous deadlines, and the proofs do not need this property either. Concretely, Theorems \ref{thm:mincost:detUB_ag} and \ref{thm:mincost:detUB} hold in the more general setting where the platform has until timestep $i + d_i$ to decide how to serve a user who arrives at timestep $i$, where the $d_i$ need not be the same. 


Another interesting observation is that knowledge of willingness to wait in this model is not very valuable for deterministic algorithms. Indeed, \texttt{Risk-Threshold-ag} is $1.61$-competitive without knowledge of $d$, and Lemma \ref{lem:mincost:detLB} shows that knowledge of $d$ can only improve the competitive ratio by $0.11$. 


We now prove Theorem~\ref{thm:mincost:detUB_ag}. The proof of Theorem~\ref{thm:mincost:detUB} is similar, but more technically involved, and can be found in Appendix~\ref{pf:thm:mincost:detUB}. 

\begin{proof}[Theorem~\ref{thm:mincost:detUB_ag}]
For a given instance $H$, let $M^*$ denote a minimum weight matching, and let $U^*$ be the set of vertices unmatched under $M^*$. For each $(i,j) \in E$, we denote $\theta_{ij} := \frac{w(i,j)}{w_i + w_j}$. For $\theta \in \bigbra{ \frac{1}{2}, 1 }$, define the set $E_\theta := \bigbrace{ (i,j) \in E : \theta_{ij} \leq \theta }$. Next, let $M$ be a matching produced by \texttt{Risk-Threshold-ag}$(H,\theta)$, and let $U$ be the set of unmatched vertices under $M$. Note that $M$ is a maximal matching in $(V,E_\theta,w)$. To compare $w(M)$ to $w(M^*)$, we assign to each vertex a score $v$ according to $M$ by the following rule:
\begin{align*}
    v_i &:= \casewise{
    \begin{tabular}{cc}
    $w_i$ & if $i \in U$ \\
    $\theta_{ij} w_i$ & if $(i,j) \in M$ for some $j \in V$ 
    \end{tabular}
    }
\end{align*}
Note that by construction, $\sum_{i \in V} v_i = w(M)$. One way to interpret $v$ is that $v_i$ is the contribution of vertex $i$ to the total cost of the matching. With this notation, we can now write
\begin{align*}
    w(M) &= \sum_{i \in V} v_i \overset{(a)}{=} \sum_{(i,j) \in M^*} (v_i + v_j) + \sum_{i \in U^*} v_i = \underbrace{\sum_{(i,j) \in M^* \cap E_\theta} (v_i + v_j)}_{\text{term 1}} + \underbrace{\sum_{(i,j) \in M^* \setminus E_\theta} (v_i + v_j)}_{\text{term 2}} + \underbrace{\sum_{i \in U^*} v_i}_{\text{term 3}},
\end{align*}
where $(a)$ is due to the fact that the set of matched nodes under $M^*$ and the set of unmatched nodes under $M^*$ form a partition of $V$. For term 1, note that $M^* \cap E_\theta$ is a matching in $(V,E_\theta,w)$. Since $M$ is a maximal matching in $(V,E_\theta,w)$, this means that for every $(i,j) \in M^* \cap E_\theta$, at least one of $i$ or $j$ is also matched in $M$. Without loss of generality, suppose $i$ is matched. Then we have 
\begin{align*}
    v_i + v_j &\leq \theta_{ij} w_i + w_j \leq \theta_{ij} \max(w_i, w_j) + \max(w_i, w_j) \leq (1+\theta) \max (w_i,w_j) \overset{(a)}{\leq} (1+\theta) w(i,j). 
\end{align*}
Where $(a)$ is due to Assumption \ref{assumption:min:split}. For term 2, note that $(i,j) \not\in E_\theta$ means that $\theta_{ij} > \theta$. Thus for any $(i,j) \in M^* \setminus E_\theta$, we have:
\begin{align*}
    v_i + v_j \leq w_i + w_j = \frac{1}{\theta_{ij}} w(i,j) \leq \frac{1}{\theta} w(i,j). 
\end{align*}
Finally, for term 3, note that $v_i \leq w_i$ is always true. Putting these observations together, we see that:
\begin{align*}
    w(M) &= \sum_{(i,j) \in M^* \cap E_\theta} (v_i + v_j) + \sum_{(i,j) \in M^* \setminus E_\theta} (v_i + v_j) + \sum_{i \in U^*} v_i \\
    &\leq \sum_{(i,j) \in M^* \cap E_\theta} (1+\theta) w(i,j) + \sum_{(i,j) \in M^* \setminus E_\theta} \frac{1}{\theta} w(i,j) + \sum_{i \in U^*} w_i \\
    &\leq \max \bigpar{ 1+\theta, \frac{1}{\theta} } \sum_{(i,j) \in M^* \cap E_\theta} w(i,j) + \max \bigpar{ 1+\theta, \frac{1}{\theta} } \sum_{(i,j) \in M^* \setminus E_\theta} w(i,j) + \max \bigpar{ 1+\theta, \frac{1}{\theta} } \sum_{i \in U^*} w_i \\
    &= \max \bigpar{ 1+\theta, \frac{1}{\theta} } \bigpar{ \sum_{(i,j) \in M^* \cap E_\theta} w(i,j) + \sum_{(i,j) \in M^* \setminus E_\theta} w(i,j) + \sum_{i \in U^*} w_i } = \max \bigpar{ 1+\theta, \frac{1}{\theta} } w(M^*). 
\end{align*}
Thus \texttt{Risk-Threshold}$(\cdot, \theta)$ is a $\max \bigpar{ 1+\theta, \frac{1}{\theta} }$-competitive algorithm. To finish, we now observe this value is minimized when $\theta = \frac{\sqrt{5}-1}{2}$, which implies that \texttt{Risk-Threshold}$\bigpar{\cdot, \frac{\sqrt{5}-1}{2} }$ is a $\frac{\sqrt{5}+1}{2}$-competitive algorithm for \onlinemin{} when $k=2$. \qed 
\end{proof}

\subsection{The $k>2$ Case}\label{sec:costmin:k>2}

This case can be addressed by a combination of ideas that we have discussed earlier. Therefore, we just state the results, sketch the main ideas, and refer the reader to the Appendix for the details. 

First, observe that the offline version of the problem \offlinemin{} is essentially a weighted set cover problem (see Appendix~\ref{sec:prelims:setcover}). To prove the equivalence, one needs to make sure Assumption \ref{assumption:min:split} does not make the problem strictly easier. This observation allows us to leverage hardness results from set cover \cite{Trevisan01} to prove hardness results for \onlinemin{}. 


We can also use the idea of batching with a random shift to get a postive result with loss of factor 2. Putting them all together gives the following: 

\begin{theorem}\label{thm:mincost:hypergraph:UB}
Let $\cA$ be an algorithm for $k$-\texttt{WSC} with approximation ratio $\rho(\cA)$. Then \texttt{Randomized-Batching}$(\cdot, \cA)$ is $\bigpar{2 - \frac{1}{d}} \rho(\cA)$-competitive for \onlinemin{}. In particular, \texttt{Randomized-}\texttt{Batching}$(\cdot, \cA)$ with $\cA =$\texttt{Weighted-Greedy} is $\bigpar{2-\frac{1}{d}}(1 + \log k)$-competitive for \onlinemin{}.

Furthermore, no algorithm with running time polynomial in $n$ and $k$ can obtain a better ratio than \onlinemin{} within $\log k - O(\log \log k)$  unless $P=NP$.
\end{theorem}

See Appendix~\ref{pf:thm:mincost:hypergraph:UB} for a proof of Theorem~\ref{thm:mincost:hypergraph:UB}.


\section{Conclusions}\label{sec:conclusion}

In this paper we studied online hypergraph matching with delays motivated by applications to high capacity ridesharing. We studied both the utility maximization and cost minimization variants of the problem. For utility maximization, we presented a polynomial-time randomized batching algorithm and proved that it attains the optimal competitive ratio. We studied cost minimization under additional assumptions specialized to the ridesharing problem. For $k=2$ we introduced a thresholding algorithm and proved that it achieves the optimal competitive ratio for deterministic algorithms. For $k>2$, we characterize both the optimal competitive ratio and the optimal polynomial-time competitive ratio up to a factor of $2 -\frac{1}{d}$. 


There are several interesting directions for future work. With regards to \onlinemax{}, adding additional application-specific assumptions could enable better approximation guarantees. For \onlinemin{}, closing the $(2-\frac{1}{d})$ gap between the achievability and impossibility results for $k>2$ and finding a randomized algorithm that beats the $\frac{3}{2}$ threshold for $k=2$ would both be interesting. For both models, extensions to heterogeneous and unobserved willingness to wait for $k>2$ would provide compelling insight for practitioners and real-world deployment. 
 

%
%
%
\bibliographystyle{splncs04} 
\bibliography{mwtsao.bib}

\ifappendix 

\newpage 

\appendix

\section{Appendix: Background}

In this section we introduce the hypergraph matching and set cover problems along with existing achievability and hardness results for these problems. Hypergraph matching and set cover play an important role in our results on utility maximization and cost minimization respectively.

\subsection{Hypergraphs}\label{sec:prelims:hypergraph}

In this section we introduce hypergraphs and the hypergraph matching problem\footnote{also known as set packing in the computer science literature}. 

\begin{definition}[Hypergraph]\label{def:hypergraph}
A hypergraph is a generalization of a graph where edges can join any number of vertices. A weighted undirected hypergraph $H = (V,E,w)$ is defined by its vertex set $V$, its hyperedge set $E$, and its hyperedge weights $w$. The hyperedge weights $w : E \rightarrow \mathbb{R}$ assign to each hyperedge a real value. For any $E' \subset E$, we define the weight of $E'$ as $w(E'):= \sum_{e \in E'} w(e)$. 
\end{definition}
\noindent Below we list some definitions related to hypergraphs that will be used throughout this paper.

\begin{definition}[Hypergraph Rank]\label{def:hypergraph:rank}
The rank of a hypergraph is the maximum cardinality of a hyperedge in the hypergraph, i.e. the hyperedges of a rank $k$ hypergraph have at most $k$ vertices. 
\end{definition}

\begin{definition}[Matching]
For a graph or hypergraph $H = (V,E)$, a matching $M$ in $H$ is a set of edges $M\subseteq E$ where all edges are disjoint. 
\end{definition}

\begin{definition}[Maximal Matching]
A matching $M$ in $H$ is maximal if it is not a strict subset of any other matching.
\end{definition}

\begin{definition}[Maximum Matching]
A matching $M$ in $H$ is maximum if it there is no matching with greater cardinality.
\end{definition}

\begin{definition}[Edge Diameter of Graphs and Hypergraphs]\label{def:edgediam}
In a graph or hypergraph $H = (V,E,w)$ with ordered vertices $V = [n]$, we define the diameter of an edge to be $\text{diam}(e) := \max_{i,j \in e} i-j$. Furthermore, we define the edge diameter of a graph $G$ to be $\text{diam}(E) := \max_{e\in E} \text{diam}(e)$, where $E$ is its edge set. 
\end{definition}

\subsubsection{Finding matchings in rank $k$ hypergraphs}

We now define the \offlinemax{} problem, which plays a key role in our discussion on online hypergraph matching.
\begin{definition}[\offlinemax{}]
The \offlinemax{} problem is to find a maximum weight matching in a rank $k$ hypergraph $H$. The maximum weight matching problem in graphs is a special case where $k=2$.
\end{definition}

\noindent The best known algorithms for \offlinemax{} \cite{Berman00,ChandraH01a} are combinatorial in nature and are $\frac{2}{k}$-competitive. Convex relaxations are surveyed in \cite{ChanLau12} who show that the standard linear programming relaxation of \offlinemax{} can be rounded to produce a $\frac{1}{k-1 + \frac{1}{k}}$-competitive algorithm. Conversely, \cite{Hazan06} proves that it is NP-hard to approximate \offlinemax{} better than $\Omega \bigpar{ \frac{\log k}{k} }$, showing that existing methods are optimal up to log factors. \\

\noindent Among all of the existing methods, we focus our discussion on a \texttt{Greedy} algorithm outlined in Algorithm~\ref{alg:offlinegreedy}. As described in Lemma \ref{lem:greedyratio}, \texttt{Greedy} is appealing because its competitive ratio is within a constant factor of the best known result and also runs in nearly linear time.

\begin{algorithm}
\caption{\texttt{Greedy}}\label{alg:offlinegreedy}
\textbf{Input:} Weighted Hypergraph $H = (V,E,w)$\;
\textbf{Output:} A matching $M$ in $H$\;
Set $M \leftarrow \emptyset$\;
Sort edges in $E=\bigbrace{e_j}_{j=1}^{|E|}$ according to their weights, in decreasing order to obtain $\bigbrace{\widetilde{e_j}}_{j=1}^{|E|}$\;
\For{$1\leq j \leq |E|$}{
\If{$M \cup \bigbrace{\widetilde{e}_j}$ is a matching}{
	$M \leftarrow M \cup \bigbrace{\widetilde{e}_j}$\;
}
}
Return $M$
\end{algorithm}

\begin{lemma}[\texttt{Greedy} is $\frac{1}{k}$-competitive]\label{lem:greedyratio}
The competitive ratio of Algorithm \ref{alg:offlinegreedy} is at least $\frac{1}{k}$ for \offlinemax{}. Furthermore, its runtime is $O \bigpar{\abs{E} \log \abs{E}}$. 
\end{lemma}


\subsection{Set Cover}\label{sec:prelims:setcover}

Crucial to our discussion on minimum cost matchings will be the set cover problem and its variants, which we define below.

\begin{definition}[Weighted Set Cover (\texttt{WSC})]\label{prob:WSC}
Given a set $X$, a collection $\cS \subset 2^X$ of subsets of $X$ satisfying $\cup_{S \in \cS} S = X$ and a weight function $w : \cS \rightarrow \mathbb{R}_+$, the weighted set cover problem \texttt{WSC}$(X,\cS,w)$ asks the following minimization problem
\begin{align*}
    \underset{\cE \subset \cS}{\text{minimize }} &\sum_{S \in \cE} w(S) \\
    \text{s.t. } & \cup_{S \in \cE} S = X. 
\end{align*}
\end{definition}

\begin{definition}[Set Cover (\texttt{SC})]\label{prob:SC}
An instance of \texttt{WSC}$(X,\cS,w)$ is an \textit{unweighted} set cover problem if $w(S) = 1$ for every $S \in \cS$. We denote such an instance as $\texttt{SC}(X,\cS)$.
\end{definition}

\begin{definition}[Set Cover with a rank constraint ($k$-\texttt{WSC})]\label{prob:kSC}
An instance \texttt{WSC}$(X,\cS,w)$ is rank $k$ if $\abs{S} \leq k$ for every $S \in \cS$. We denote the rank $k$ set cover problem as $k$-\texttt{WSC}$(X,\cS,w)$. If in addition, $w(S) = 1$ for every $S \in \cS$, then this is an instance of rank $k$ unweighted set cover, which we denote as $k$-\texttt{SC}$(X,\cS)$.  
\end{definition}

\noindent The set cover problem and its variants have been extensively studied in the computer science literature. A simple greedy algorithm was shown in \cite{Johnson74,Lovasz75} to be $1 + \log n$ competitive for \texttt{SC}, where $n = \abs{X}$. It was later shown by \cite{Feige98} that this performance is optimal up to lower order terms. Namely, they showed that a polynomial-time $(1-\epsilon) \log n$ competitive ratio is achievable only if NP has quasi-polynomial-time algorithms. \\

\noindent For $k$-\texttt{SC}, the greedy algorithm presented in \cite{Johnson74,Lovasz75} is $1 + \log k$ competitive. \cite{Chvatal79} extended this result to the weighted version of the problem $k$-\texttt{WSC}. Trevisan showed that this performance is optimal up to lower order terms in \cite{Trevisan01} by showing that it is NP-hard to approximate $k$-\texttt{SC} better than $\log k - O(\log \log k)$. 

\newpage 
\section{Appendix: Proofs for Utility Maximization}

\subsection{Visualizing the reduction from $d$-\texttt{ASP} to \onlinemax{}}\label{dASP_vis}
\begin{figure}[h]
\centerline{
	\includegraphics[width = 1.0\textwidth]{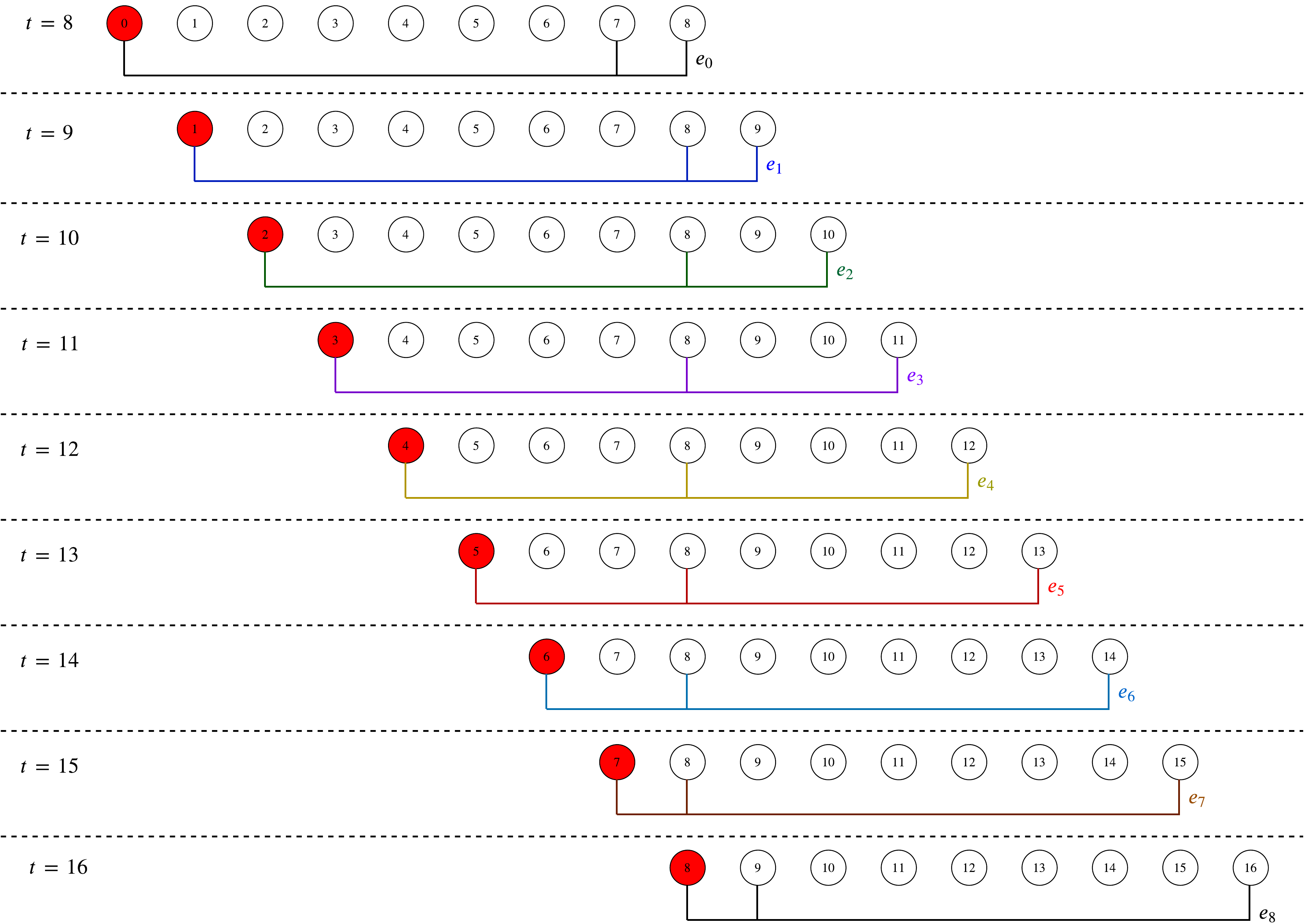}
}
\caption{An example of a shareability hypergraph $G \in \cG$ evolving in time. This example has $d=9$. Vertices arrive one per timestep on the right, and depart one per timestep from the left. A vertex is red during the last timestep it is in the system. For each $t$, $e_t$ is revealed at time $t+d-1$, when its last vertex appears. It also becomes critical at this time, since $t \in e_t$ is about to leave the system. Therefore, from timestep $d-1$ onward, one edge, along with its weight, is revealed each timestep, and can only be included into the matching at that timestep. Furthermore, at most one edge can be in the matching, since all edges contain the node $d-1=8$. Hence, the family $\cG$ encodes the \texttt{$d$-ASP} problem.}
\label{fig:d-ASP}
\end{figure}

\subsection{Proof of Lemma \ref{lem:advsecretary} }\label{pf:lem:advsecretary}
For any (possibly randomized) algorithm $\cA$, we use the random variable $W(\cA)$ to denote the aptitude of the candidate chosen by $\cA$. Let $\cA_t$ be any algorithm with the optimal competitive ratio $\rho(\cA_t)$ for \texttt{$t$-ASP}. First, $\rho(\cA_t) \geq \frac{1}{t}$ by the trivial algorithm $\cA_t^*$ which ignores all scores, and commits to one of the $t$ secretaries uniformly at random. Indeed,
\begin{align*}
\mathbb{E} \bigbra{ W(\cA_{t}^*) } &= \frac{1}{t} \sum_{\tau = 1}^t w_\tau = \frac{\norm{w}_1}{t} \geq \frac{\norm{w}_\infty}{t} = \frac{\texttt{OPT}}{t}. 
\end{align*}
Therefore, all we need to show is that $\rho(\cA_t) \leq \frac{1}{t}$, which we will do by induction. \\

\noindent \textit{Base Case:} For $t = 1$ the problem is trivial; there is only one candidate, so hiring the candidate gives the optimal solution and the optimal competitive ratio is $\frac{1}{1} = \frac{1}{t}$. \\

\noindent \textit{Induction Step:} Suppose that $\rho(\cA_t) = \frac{1}{t}$ for some $t \geq 1$. For the $(t+1)\texttt{-ASP}$, at the first timestep, $\cA_{t+1}$ decides to hire the first candidate with some probability $p$. If the first candidate is hired, the algorithm collects utility $w_1$. If the first candidate is not hired, then the algorithm must solve a secretary problem on the remaining $t$ candidates. Therefore, 
\begin{align*}
\mathbb{E} \bigbra{ W(\cA_{t+1}) } = p w_1 + (1-p) \mathbb{E} \bigbra{W(\cA_t)}.
\end{align*}
Therefore the competitive ratio of $\cA_{t+1}$ is
\begin{align}\label{eqn:secretaryinduction}
\rho(\cA_{t+1}) = \min_{w \in \mathbb{R}^{t+1}_+} \frac{p w_1}{\norm{w}_\infty} + (1-p)\frac{\mathbb{E}[W(\cA_{t})]}{\norm{w}_\infty}
\end{align}
Let $z \in \mathbb{R}^{t}_+$ be a worst-case instance for $\cA_{t}$ so that $\mathbb{E}[W(\cA_t)] = \frac{\texttt{OPT}}{t}$. Now let $[w_2,w_3,...,w_{t+1}]^\top = \alpha z$ for some $\alpha > 0$. Sending $w_1 \rightarrow \infty$, \eqref{eqn:secretaryinduction} converges to $p$. Sending $\alpha \rightarrow \infty$, eventually we have $\norm{w}_\infty = \norm{z}_\infty$, and \eqref{eqn:secretaryinduction} converges to $(1-p)\rho(\cA_t) = \frac{1-p}{t}$. Therefore,
\begin{align*}
\rho(\cA_{t+1}) &\leq \min \bigpar{ p, \frac{1-p}{t} }.
\end{align*}
The best choice of $p$ to maximize this upper bound is $p = \frac{1}{t+1}$, which yields $\rho(\cA_{t+1}) = \frac{1}{t+1}$. This completes the induction step, and thus completes the proof. 

\subsection{Proof of Theorem \ref{thm:online_polyLB} }\label{pf:thm:online_polyLB}

Let $M^*$ be a maximum weight matching for an instance $H = (V,E,w)$ of \onlinemax{}. Partition $M^*$ into two disjoint subsets $M_{s}^* \cup M_{\ell}^*$ defined via
\begin{align*}
M_{s}^* &:= \bigbrace{e \in M^* : \text{diam}(e) \leq d-k} \\
M_{\ell}^* &:= \bigbrace{e \in M^* : \text{diam}(e) > d-k},
\end{align*} 
where $M_{s}^*$ contains all of the edges of $M^*$ whose diameter is at most $d-k$, and $M_{\ell}^*$ contain the other edges. First note that $w(M^*) = w(M_{s}^*) + w(M_{\ell}^*)$. \\

\noindent Recall the definition of batches: For each $z \in \bigbrace{0,1,...,d-1}$ and $i \leq \frac{n}{d}$, we define the batch $W_{z,i} := (V_{z,i}, E_{z,i},w)$ to be a subgraph of $H$ in the following way:
\begin{enumerate}
	\item $V_{z,i} = \bigbrace{z + id + t}_{t=0}^{d-1}$ is the set of nodes that arrive between times $z+id$ and $z+(i+1)d-1$. 
	\item $E_{z,i} \subset E$ are all edges whose endpoints are a subset of $V_{z,i}$. 
\end{enumerate}
We then use $E_z := \bigcup_{i} E_{z,i}$. 

\noindent Now suppose we have a procedure $\cA$ that, for every pair $z,i$, can find a matching $M_{z,i} \subset E_{z,i}$ so that
\begin{align}\label{eqn:hybrid_ratio}
w(M_{z,i}) \geq w(M_\ell^* \cap E_{z,i}) + \frac{1}{k} w(M_{s}^* \cap E_{z,i}).
\end{align}
Next we show that Algorithm \ref{alg:randombatch} using such an $\cA$ as a subroutine is $\frac{1}{d}$-competitive for \onlinemax{}. 
\begin{align}
\mathbb{E} \bigbra{ w(M) } &= \sum_{z=0}^{d-1} \mathbb{P} \bigbra{Z=z} \mathbb{E} \bigbra{ w(M) | Z = z} \label{eqn: 1/dpoly}\\
&= \frac{1}{d} \sum_{z=0}^{d-1} \sum_{i=0}^{n/d} w(M_{z,i}) \nonumber \\
&\overset{(a)}{\geq} \frac{1}{d} \sum_{z=0}^{d-1} \sum_{i=0}^{n/d} \bigbra{ w(M_\ell^* \cap E_{z,i}) + \frac{1}{k} w(M_{s}^* \cap E_{z,i}) } \nonumber \\
&= \frac{1}{d} \sum_{z=0}^{d-1} \bigbra{ w(M_\ell^* \cap E_{z}) + \frac{1}{k} w(M_{s}^* \cap E_{z}) } \nonumber \\
&= \frac{1}{d} \sum_{z=0}^{d-1} \sum_{e \in M_\ell^* \cap E_{z}} w(e) + \frac{1}{kd} \sum_{z=0}^{d-1} \sum_{e \in M_s^* \cap E_{z}} w(e) \nonumber \\
&= \frac{1}{d} \sum_{e \in M_\ell^*} w(e) \sum_{z=0}^{d-1} \mathds{1}_{[e \in E_z]} + \frac{1}{kd} \sum_{e \in M_s^*} w(e) \sum_{z=0}^{d-1} \mathds{1}_{[e \in E_z]} \nonumber \\
&\overset{(b)}{\geq} \frac{1}{d} \sum_{e \in M_\ell^*} w(e) + \frac{1}{kd} \sum_{e \in M_s^*} w(e) k = \frac{1}{d} \sum_{e \in M^*} w(e) = \frac{w(M^*)}{d}. \nonumber 
\end{align}
Where $(a)$ is due to condition \eqref{eqn:hybrid_ratio} and $(b)$ is due to Observation \ref{obs:indicatorsum}; edges in $M_s^*$ have diameter at most $d-k$ and edges in $M_\ell^*$ have diameter at most $d-1$. Therefore if we can find a $\cA$ that runs in polynomial-time and satisfies \eqref{eqn:hybrid_ratio} for all pairs $z,i$, then Algorithm \ref{alg:randombatch} using $\cA$ as a subroutine will achieve a competitive ratio of $\frac{1}{d}$ in expectation. 

\subsubsection{Finding a subroutine that satisfies \eqref{eqn:hybrid_ratio}:}

Let $W_{z,i} = (V_{z,i},E_{z,i},w)$ be a batch. Without loss of generality (by shifting the time indices) we will assume that $V_{z,i} = \bigbrace{0,1,...,d-1}$, i.e. $z=i=0$. Recalling that \texttt{Greedy} is $\frac{1}{k}$-competitive for \offlinemax{}, for $M = \texttt{Greedy}(W_{z,i})$ we have
\begin{align}
w(M) &\geq \frac{1}{k} w(M^* \cap E_{z,i}) \nonumber \\
&= \frac{1}{k} w(M_\ell^* \cap E_{z,i}) + \frac{1}{k} w(M_s^* \cap E_{z,i}).\label{eqn:greedy_short_long}
\end{align}
Following the same analysis as \eqref{eqn: 1/dpoly}, running Algorithm \ref{alg:randombatch} with \texttt{Greedy} gives
\begin{align}
\mathbb{E} \bigbra{w(M)} &\geq \frac{1}{kd} \bigpar{ \sum_{e \in M_\ell^*} w(e) \sum_{z=0}^{d-1} \mathds{1}_{[e \in E_z]} + \sum_{e \in M_s^*} w(e) \sum_{z=0}^{d-1} \mathds{1}_{[e \in E_z]} } = \frac{1}{kd} w(M_\ell^*) + \frac{1}{d}  w(M_s^*).\label{eqn:rb_greedy_short_long} 
\end{align}

This is not sufficient for our purposes because the $\frac{1}{k}$ coefficient in front of $w(M_\ell^* \cap E_{z,i})$ in \eqref{eqn:greedy_short_long} leads to a $\frac{1}{kd} w(M_\ell^*)$ contribution to the bound in \eqref{eqn:rb_greedy_short_long}. Note that the reason we are able to tolerate the $\frac{1}{k}$ coefficient on $w(M_s^* \cap E_{z,i})$ in \eqref{eqn:hybrid_ratio} is because of Observation \ref{obs:indicatorsum}; each $e \in M_s^*$ is in $E_z$ for at least $k$ choices of $z$, which will cancel the $\frac{1}{k}$ coefficient. However, edges in $M_\ell^*$ are in $E_z$ for less than $k$ choices of $z$, and hence we cannot do the same cancellation. Thus, we need to find a way to modify \texttt{Greedy} to get a better approximation factor for $w(M_\ell^* \cap E_{z,i})$. \\

\noindent To this end, we have the following observation. Let $L := \bigbrace{0,1,..,k-1}$ and $R := \bigbrace{d-k,d-k+1,...,d-1}$ be the sets of the first and last $k$ vertices of $W_{z,i}$ respectively. 

\begin{observation}
For $e \in M^* \cap E_{z,i}$, if $e \cap L = \emptyset$ or $e \cap R = \emptyset$, then $diam(e) \leq d-k$, and hence $e \not\in M_\ell^*$. Equivalently, for any $e \in M_\ell^*$, we have $e \cap L \neq \emptyset$ and $e \cap R \neq \emptyset$.
\end{observation}

\noindent Hence $M_\ell^* \cap E_{z,i}$ is a member of the following set: 
\begin{align*}
\cM_{L,R}(W_{z,i}) := \bigbrace{M \text{ a matching in } W_{z,i} : e \cap L \neq \emptyset, e \cap R \neq \emptyset \; \forall \; e \in M}.
\end{align*}

\noindent In other words, $M_\ell^* \cap E_{z,i}$ is a matching in $W_{z,i}$ that is restricted to $L,R$ where we define restricted matchings below. 

\begin{definition}[Restricted Matchings]\label{def:restricted_matchings}
Given a hypergraph $H = (V,E)$ and two disjoint subsets $L,R \subset V, L \cap R = \emptyset$, a collection of edges $M \subset E$ is a matching is restricted to $L,R$ if 
\begin{enumerate}
	\item $M$ is a matching in $H$.
	\item For any $e \in M$, $e \cap L \neq \emptyset$ and $e \cap R \neq \emptyset$.
\end{enumerate}
Finally, we use $\cM_{L,R}(H)$ to denote the set of all matchings in $H$ that are restricted to $L,R$. 
\end{definition}

Since $\abs{L} = \abs{R} = k$ is a constant, $\abs{\cM_{L,R}(W_{z,i})}$ is small (polynomial in $d$). This motivates the \hyperref[alg:dk_greedy]{\texttt{Depth-$k$ Greedy}} algorithm which does the following: For each $M \in \cM_{L,R}(W_{z,i})$, extend $M$ to a maximal matching by adding edges via \texttt{Greedy}, and then among the $\abs{\cM_{L,R}(W_{z,i})}$ resulting maximal matchings, return the one with the largest weight. First, we will show that \texttt{Depth-$k$ Greedy} satisfies \eqref{eqn:hybrid_ratio}. To this end, for any matching $M$ we use $V(M)$ to denote the vertices that are matched by $M$. If $M_{z,i}$ is the matching returned by \texttt{Depth-$k$ Greedy}, we have
\begin{align*}
w(M_{z,i}) &= \max_{M \in \cM_{L,R}(W_{z,i})} w(M) + w(\texttt{Greedy}(V_{z,i} \setminus V(M), E_{z,i})) \\
&\overset{(a)}{\geq} w(M_\ell^* \cap E_{z,i}) + w(\texttt{Greedy}(V_{z,i} \setminus V(M_\ell^* \cap E_{z,i}), E_{z,i})) \\
&\overset{(b)}{\geq} w(M_\ell^* \cap E_{z,i}) + \frac{1}{k} \max_{\substack{M' \text{ s.t.} \\ M' \cup (M_\ell^* \cap E_{z,i}) \text{ is}\\ \text{a matching} \\ \text{in } W_{z,i}}} w(M') \\
&\overset{(c)}{\geq} w(M_\ell^* \cap E_{z,i}) + \frac{1}{k} w(M_s^* \cap E_{z,i}).
\end{align*}
Where $(a)$ is due to $M_\ell^* \cap E_{z,i} \in \cM_{L,R}(W_{z,i})$, $(b)$ is due to \texttt{Greedy} being $\frac{1}{k}$-competitive for \offlinemax{}, and $(c)$ is because $(M_{\ell}^* \cap E_{z,i}) \cup (M_{s}^* \cap E_{z,i}) = M^* \cap E_{z,i}$ is a matching in $W_{z,i}$. Hence \texttt{Depth-$k$ Greedy} satisfies condition \eqref{eqn:hybrid_ratio}. To complete the proof, we next show that the runtime of \texttt{Depth-$k$ Greedy} is polynomial in the size of the problem instance $H$.  

\subsubsection{Runtime of \texttt{Depth-$k$ Randomized Batching:}}
Finally, we need to show that the runtime of \texttt{Depth-$k$ Randomized Batching} is polynomial in the size of $H$. Since $\abs{L} = k$, any matching in $\cM_{L,R}(W_{u,i})$ can have at most $k$ hyperedges. From this, we conclude that
\begin{align*}
\abs{\cM_{L,R}(W_{z,i})} &\leq \sum_{\ell = 1}^{k} {\abs{E_{z,i}} \choose \ell} \leq k {\abs{E_{z,i}} \choose k} \leq k \bigpar{ \frac{\abs{E_{z,i}}}{ek} }^k
\end{align*}
For a batch $W_{z,i}$, it takes $c \abs{E_{z,i}} \log \abs{E_{z,i}} $ time to sort the edges by weight, for some constant $c$. For each $M \in \cM_{L,R}(W_{z,i})$, given that the edges are already sorted, extending $M$ to a maximal matching via \hyperref[alg:offlinegreedy]{\texttt{Greedy}} requires $O(k \abs{E_{z,i}})$ time, since for each edge, it takes $O(k)$ time to check if it is disjoint from the current matching. Doing this for all members of $\cM_{L,R}(W_{z,i})$ would take
\begin{align*}
c \abs{E_{z,i}} \log \abs{E_{z,i}} + \abs{\cM_{L,R}(W_{z,i})} \cdot k \abs{E} &= c \abs{E_{z,i}} \log \abs{E_{z,i}} + \frac{k^2}{(ek)^k} \abs{E_{z,i}}^{k+1} 
\end{align*}
time. Therefore the expected runtime of \texttt{Randomized-Batching}$(\cdot,\cA)$ with $\cA = \texttt{Depth-k Greedy}$ is
\begin{align*}
\frac{1}{d} \sum_{z=0}^{d-1} \sum_{i} \bigpar{ c \abs{E_{z,i}} \log \abs{E_{z,i}} + \frac{k^2}{(ek)^k} \abs{E_{z,i}}^{k+1} } &= O \bigpar{ \frac{1}{d} \sum_{z=0}^{d-1} \sum_{i} \abs{E_{z,i}}^{k+1} }\\
&\leq O \bigpar{ \frac{1}{d} \sum_{z=0}^{d-1} \abs{E_{z}}^{k+1} }  \\
&\leq O \bigpar{ \max_{0 \leq z \leq d-1} \abs{E_{z}}^{k+1} } \\
&\leq O \bigpar{ \abs{E}^{k+1} }.
\end{align*}
This completes the proof.

\section{Appendix: Proofs for Cost Minimization in Ridesharing}

\subsection{Proof of Lemma \ref{lem:mincost:detLB}}\label{pf:lem:mincost:detLB}

Consider the following pair of shareability graphs $H_1,H_2$, where the arrival order of vertices is $A,B,C,D$, the cost of each vertex is $1$, and the cost of every edge is also $1$. The minimum weight matchings are highlighted for each graph. \\
\centerline{
\begin{tabular}{cc}
    \includegraphics[scale = 0.5]{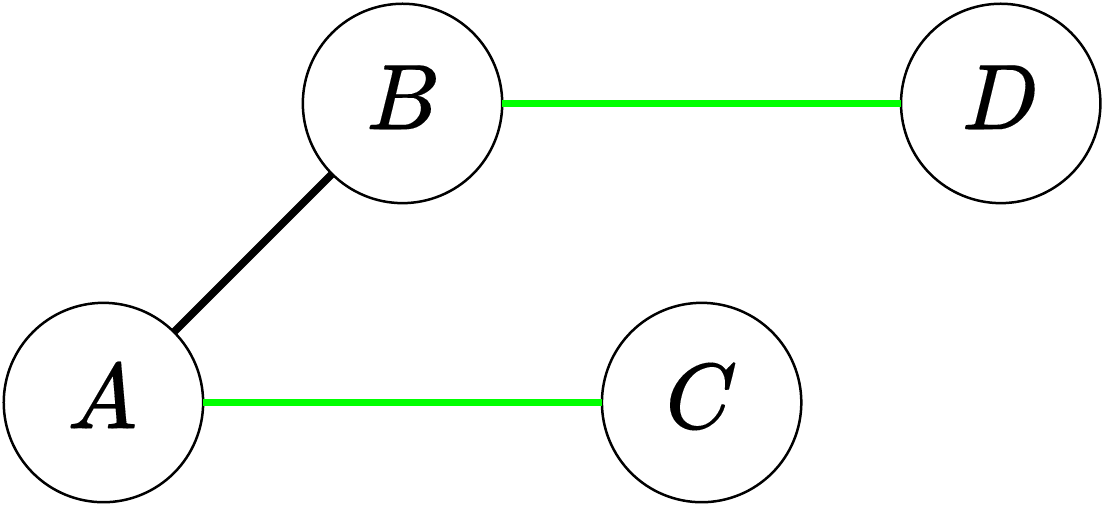} & \includegraphics[scale = 0.5]{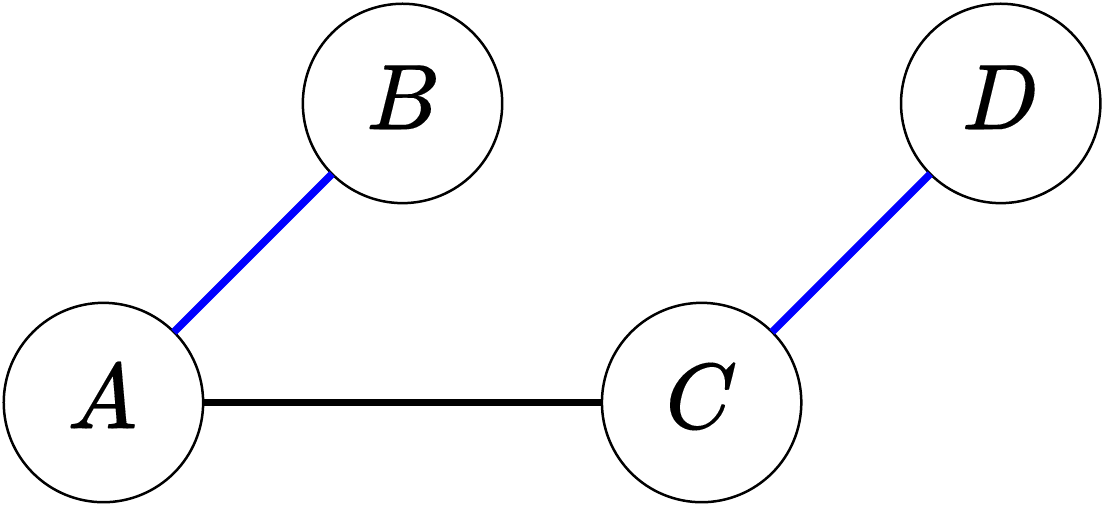} \\
    $H_1$ & $H_2$
\end{tabular}
}

When $A$ becomes critical, there are only 3 possible actions for a deterministic algorithm: 
\begin{itemize}
    \item[(a)] Pair $(A,B)$.
    \item[(b)] Pair $(A,C)$.
    \item[(c)] Don't pair $A$ with anything. 
\end{itemize}

\subsubsection{Proving a $\frac{3}{2}$ lower bound for deterministic algorithms}

Action $(a)$ is optimal for $H_2$, but leads to a non-maximum matching for $H_1$. On the other hand, action $(b)$ is optimal for $H_1$ but leads to a non-maximum matching for $H_2$. \noindent A non-maximum matching has either 0 edges and 4 unmatched vertices, which gives a cost of $4$, or 1 edge and 2 unmatched vertices, which gives a cost of $3$. A maximum matching has 2 edges, and a cost of $2$. Therefore, every non-maximum matching has cost at least $\frac{3}{2}$ times as large as an optimal matching. Before $D$ arrives, the instances $H_1,H_2$ look identical: \\
\centerline{
    \includegraphics[scale = 0.5]{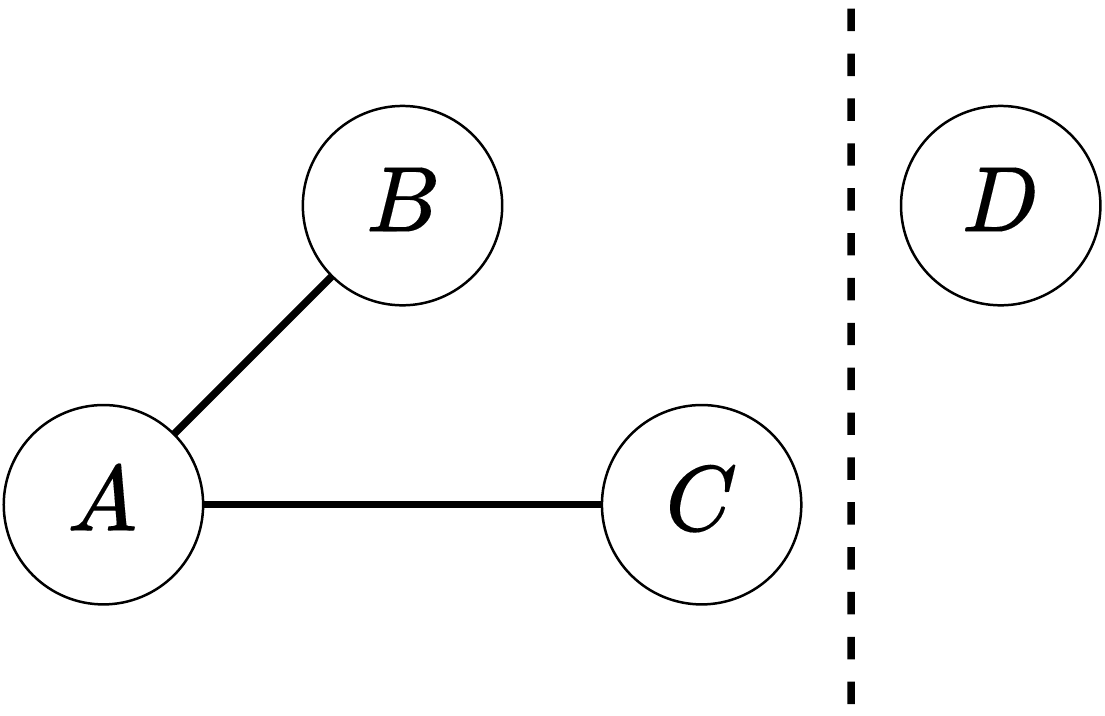}
}
Therefore, if $A$ becomes critical before $D$ arrives, then there is no way to distinguish whether the instance is $H_1$ or $H_2$. Using action $(a)$ will lead to a non-maximum matching if the instance is $H_1$, and using the action $(b)$ leads to a non-maximum matching if the instance is $H_2$. Action $(c)$ always leads to a non-maximum matching. \\ 

\noindent Now we review the observations to prove a lower bound of $\frac{3}{2}$ for all deterministic algorithms. First, on the instances $\bigbrace{H_1,H_2}$, there are only 3 choices for the first action in any deterministic algorithm. Next, for each of the three possible first actions, there is an instance $H \in \bigbrace{H_1,H_2}$ for which this action leads to a non-maximum matching. Finally, for any $H \in \bigbrace{H_1,H_2}$, a non-maximum matching has cost at least $\frac{3}{2}$ times as large as the minimum cost matching. Therefore, for any deterministic algorithm, there is an instance for which the algorithm's cost is at least $\frac{3}{2}$ times that of the minimum cost matching.  

\subsubsection{Proving a $\frac{5}{4}$ lower bound for randomized algorithms}

We will prove a $\frac{5}{4}$ lower bound using Yao's minimax principle. When $A$ becomes critical, the only actions for deterministic algorithms are $\bigbrace{(a),(b),(c)}$. Thus the behavior of any randomized algorithm at the time $A$ becomes critical can be described as a distribution over $\bigbrace{(a),(b),(c)}$. Let $p_a,p_b,p_c$ be the probabilities that actions $(a),(b),(c)$ are taken respectively. Since $p_a + p_b + p_c = 1$, at least one of the following must happen:
\begin{itemize}
    \item[]Case 1: $p_a + p_c \geq \frac{1}{2}$. 
    \item[]Case 2: $p_b + p_c \geq \frac{1}{2}$.  
\end{itemize}
Note that actions $(a),(b)$ lead to non-maximum matchings for the instance $H_1$, and actions $(b),(c)$ lead to non-maximum matchings for the instance $H_2$. If Case 1 is true, then on the instance $H_1$, and letting $M$ be the matching produced by the randomized algorithm,
\begin{align*}
    \mathbb{E}\bigbra{c(M)} &= p_b \mathbb{E} \bigbra{M | \text{first action is } (b)} + (p_a + p_c) \mathbb{E} \bigbra{M | \text{first action } \in \bigbrace{(a),(c)}} \\
    &\geq p_b \cdot 2 + (p_a + p_c) 3 \\
    &\geq 1 + \frac{3}{2} = \frac{5}{2}. 
\end{align*}
Where the first inequality is because $2$ is the minimum cost matching, so in particular it is a lower bound for $\mathbb{E} \bigbra{M | \text{first action is } (b)}$, and because taking action $(a)$ or $(c)$
leads to a non-maximum matching in $H_1$, which has cost at least $3$. Therefore, $\frac{c(M)}{c(M^*)} \geq \frac{5}{2} \div 2 = \frac{5}{4}$. The analysis if Case 2 is true is analogous; the only difference is we apply the analysis to the $H_2$ instance instead. 

\subsection{Proof of Theorem \ref{thm:mincost:detUB}}\label{pf:thm:mincost:detUB}

For an instance $H = (V,E,w)$, let $M^*$ be a minimum weight matching for $H$, and let $U^*$ denote the set of vertices that are unmatched under $M^*$. For each edge $(i,j) \in E$, define $\theta_{ij} := \frac{w(i,j)}{w_i + w_j}$. For $\theta \in \bigbra{\frac{1}{2},1}$, define
\begin{align*}
    E_\theta := \bigbrace{ (i,j) \in E : \theta_{ij} \leq \theta }.
\end{align*}
Let $M$ be the matching produced by \texttt{Risk-Threshold}$(H, \theta)$, and let $U$ be the set of vertices that are unmatched under $M$. Note that $M$ is a maximal matching in $(V, E_\theta, w)$. Define $v^*,v$ in the following way:
\begin{align*}
    v_i^* := \casewise{
        \begin{tabular}{cc}
            $w_i$ & if $i \in U^*$  \\
            $\theta_{ij} w_i$ & if $(i,j) \in M^*$ 
        \end{tabular}
    }
    \text{ and }
    v_i := \casewise{
        \begin{tabular}{cc}
            $w_i$ & if $i \in U$  \\
            $\theta_{ij} w_i$ & if $(i,j) \in M$.
        \end{tabular}
    }
\end{align*}
Note that by construction, $w(M^*) = \sum_{i \in V} v_i^*$ and $w(M) = \sum_{i \in V} v_i$, so $v_i, v_i^*$ denote the contribution of vertex $i$ to the total cost under matchings $M,M^*$ respectively. \\

\noindent Since $M,M^*$ are matchings, their union $M \cup M^*$ is a disjoint union of cycles and paths, i.e. every connected component of $(V,M \cup M^*,w)$ falls into one of the three following categories:
\begin{enumerate}
    \item Vertices of degree 0.
    \item Cycles.
    \item Maximal Paths.
\end{enumerate}
Where the definition of a maximal path is the following. 
\begin{definition}[Maximal Paths]
We say a path $P$ is maximal in a graph $G = (V,E,w)$ if it is not a strict subset of another path. 
\end{definition}
To prove Theorem \ref{thm:mincost:detUB}, we will prove the following claim.
\begin{observation}\label{obs:detUB}
For every connected component in $M \cup M^*$, the cost of the vertices in this connected component under $M$ is no more than $\frac{3}{2}$ times the cost of those same vertices under $M^*$. 
\end{observation}
\noindent Once we establish this, Theorem \ref{thm:mincost:detUB} follows by the following implications. Suppose 
\begin{align*}
    M \cup M^* = \bigpar{ \bigcup_{j=1}^p C_j } \bigcup \bigpar{ \bigcup_{\ell=1}^q P_\ell }
\end{align*}
is the partition of $M \cup M^*$ into cycles $\bigbrace{C_j}_{j=1}^p$ and maximal paths $\bigbrace{P_\ell}_{\ell=1}^q$. If for a collection of edges $E'$ we use $V(E')$ to denote the set of vertices that are touched by edges in $E'$, we have 
\begin{align*}
    V = \bigpar{ \bigcup_{j=1}^p V(C_j) } \bigcup \bigpar{ \bigcup_{\ell=1}^q V(P_\ell) } \bigcup (U \cap U^*)
\end{align*}
where recall that $U \cap U^*$ denotes the vertices that are unmatched under both $M$ and $M^*$. Finally we can see that:
\begin{align*}
    w(M) &= \sum_{i \in V} v_i \\
    &= \sum_{j=1}^p \sum_{i \in V(C_j)} v_i + \sum_{\ell=1}^q \sum_{i \in V(P_\ell)} v_i + \sum_{i \in U \cap U^*} v_i \\
    &\overset{(a)}{\leq} \sum_{j=1}^p \frac{3}{2} \sum_{i \in V(C_j)} v_i^* + \sum_{\ell=1}^q \frac{3}{2} \sum_{i \in V(P_\ell)} v_i^* + \sum_{i \in U \cap U^*} v_i^* \\ 
    &\leq \frac{3}{2} \bigpar{ \sum_{j=1}^p \sum_{i \in V(C_j)} v_i^* + \sum_{\ell=1}^q \sum_{i \in V(P_\ell)} v_i^* + \sum_{i \in U \cap U^*} v_i^* } \\
    &= \frac{3}{2} \sum_{i \in V} v_i^* = \frac{3}{2} w(M^*).
\end{align*}
Where $(a)$ is due to Observation \ref{obs:detUB}. We will prove claim 3 by examining each of the 3 types of connected components of $(V,M \cup M^*, w)$. 

\subsubsection{Vertices of degree 0}

If a vertex $i$ has degree zero in the graph $(V,M\cup M^*)$, this means it is unmatched under both $M$ and $M^*$, meaning that $v_i = w_i$ and $v_i^* = w_i$, therefore the cost contribution of $i$ to $M$ is the same as its cost contribution to $M^*$.

\subsubsection{Cycles}\label{pf:mincost:detUB:cycles}

Note that any cycle in $(V,M \cup M^*,w)$ must have even length. A cycle having odd length would imply that there are two edges from the same matching that share a vertex. So consider an even cycle on vertices $\bigbrace{1,2,...,2\ell}$ where the edges $\bigbrace{(2k+1,2k+2)}_{k=0}^{\ell-1}$ belong to $M^*$ and the edges $\bigbrace{(2k, 2k+1)}_{k=1}^{\ell}$ belong to $M$. Since every vertex in a cycle has degree 2, this means that each vertex in the cycle is matched in both $M$ and $M^*$. We thus see that

\begin{align*}
    \sum_{i=1}^{2\ell} v_i = \sum_{k=1}^{\ell} w(2k,2k+1) &\overset{(a)}{\leq} \sum_{k=1}^{\ell} \theta w_{2k}+ \theta w_{2k+1} \\
    &\overset{(b)}{\leq} \sum_{k=1}^{\ell} 2\theta v_{2k}^* + 2\theta v_{2k+1}^* = 2\theta \sum_{i=1}^{2\ell} v_i^*.
\end{align*}
Where $(a)$ is due to the fact that $w(i,j) = \theta_{ij}(w_i+w_j)$, $\theta_e \leq \theta$ for every $e \in E_\theta$ and $M \subset E_\theta$. We get $(b)$ because $\theta \geq \frac{1}{2}$ due to Assumption \ref{assumption:min:split}. We conclude this part by noting that for $\theta = \frac{2}{3}$, $2\theta = \frac{4}{3} \leq \frac{3}{2}$. \\

\subsubsection{Paths}
To analyze the cost contributions of maximal paths, we will use the following definition of restricted cost. 
\begin{definition}[Restricted Cost]
Given a graph $G = (V,E,w)$, an edge sets $E'\subset E$, and a matching $M$, we define $w(M ; E')$ to be the cost of $M \cap E'$ in the subgraph $(V(E'), E', w)$. 
\end{definition}
To prove Observation \ref{obs:detUB} for paths, we will use the following lemma. 
\begin{lemma}[Path cost bound]\label{lem:mincost:pathbound_master}
For any path $P \subset (V, M \cup M^*,w)$, we have:
\begin{align*}
    w(M;P) \leq \frac{3}{2} w(M^*;P). 
\end{align*}
\end{lemma}
\noindent See Section \ref{pf:lem:mincost:pathbound_master} for a proof of Lemma \ref{lem:mincost:pathbound_master}. Once we establish this lemma, we can establish Observation \ref{obs:detUB} for maximal paths using the following observation.

\begin{observation}\label{obs:w_and_v}
If $P$ is a maximal path, $\sum_{i \in V(P)} v_i = w(M ; P)$ and $\sum_{i \in V(P)} v_i^* = w(M^* ; P)$. 
\end{observation}

\begin{proof}[proof of Observation \ref{obs:w_and_v}]
Note that $\sum_{i \in V(P)} v_i \neq w(M ; P)$ only if there is some vertex $i \in V(P)$ so that either
\begin{enumerate}
    \item $i$ is matched under $M \cap P$ in $(V(P),P,w)$ but is not matched under $M$ in $(V,E,w)$, or
    \item $i$ is not matched under $M \cap P$ in $(V(P),P,w)$ but is matched under $M$ in $(V,E,w)$.
\end{enumerate}
The first situation can never happen because $M \cap P \subset M$, so every vertex matched by $M \cap P$ is also matched by $M$. If the second situation happens, this means that the degree of $i$ is 1 in $(V(P),P,w)$ and there exists some $i' \in V \setminus V(P)$ with $(i,i')\in M$. But this implies that $P \cup \bigbrace{(i,i')}$ is a path in $M \cup M^*$, which contradicts the assumption that $P$ is maximal. The proof for $M^*$ is analogous.  
\end{proof}

\noindent Using Observation \ref{obs:w_and_v}, we can prove Observation \ref{obs:detUB} for maximal paths $P$ via
\begin{align*}
    \sum_{i \in V(P)} v_i = w(M ; P) \leq \frac{3}{2} w(M^*; P) = \sum_{i \in V(P)} v_i^*. 
\end{align*}
This completes the proof for Theorem \ref{thm:mincost:detUB}. 

\subsection{Proof of Lemma \ref{lem:mincost:pathbound_master}}\label{pf:lem:mincost:pathbound_master}

We partition the set of maximal paths in $(V, M \cup M^*,w)$ into four categories, and treat each of the categories separately. 
\begin{enumerate}
    \item Length 1 paths. 
    \item Type 1 paths.
    \item Even length paths.
    \item Type 2 paths of length at least 3. 
\end{enumerate}

\noindent Where type 1 and type 2 paths are defined as follows: 
\begin{definition}[Path Types]
A path in $(V, M \cup M^*,w)$ is type 1 if it has an odd length, and contains more edges from $M$ than from $M^*$. A path is type 2 if it has odd length and contains more edges from $M^*$ than from $M$. 
\end{definition}

Proving Lemma \ref{lem:mincost:pathbound_master} for paths in categories 1 and 2 are easy, so we will do this directly. For Categories 2 and 3, we will prove the lemma by induction on the path length. 

\subsubsection{Maximal paths of length 1}\label{pf:thm:mincost:detUB:len1}
Consider a path of length 1 in $M \cup M^*$, i.e. an edge $e = (i,j) \in E$. There are two cases to consider: 1) $e \in M$ and 2) $e \in M^* \setminus M$. \\

\noindent \textit{Case 1:} First note that either $(i,j) \in M^*$, or both $i,j$ are unmatched under $M^*$. To see this, if $(i,j') \in M^*$ for some $j\neq j'$, then the path $\bigbrace{(j,i)}$ would not be maximal since in particular the path $\bigbrace{(j,i), (i,j')} \subset M \cup M^*$. From this we can conclude that that $v_i^* \geq \theta_{ij} w_i, v_j^* \geq \theta_{ij} w_j^*$. Finally,
\begin{align*}
    v_i + v_j = w(i,j) = \theta_{ij} w_i + \theta_{ij} w_j \leq v_i^* + v_j^*. 
\end{align*}

\noindent \textit{Case 2:} Now we consider the case where $e \in M^* \setminus M$. We argue by contradiction that $e \not \in E_\theta$. Suppose $e \in E_\theta$. We can immediately say that neither $i$ nor $j$ matched under $M$, because if any of the vertices of $e$ are matched under $M$, then this path would have more than $1$ edge. Since none of the vertices of $e$ are matched in $M$, then $M \cup \bigbrace{e}$ is a matching in $(V,E_\theta)$. However, this contradicts the fact that $M$ is a maximal matching in $(V,E_\theta)$. Therefore if $e \in M^*$, it must be the case that $e \not\in E_\theta$. From this we see that:
\begin{align*}
    v_i + v_j \leq w_i + w_j = \frac{w(i,j)}{\theta_{ij}} = \frac{v_i^* + v_j^*}{\theta_{ij}} \leq \frac{v_i^* + v_j^*}{\theta} = \frac{3}{2} (v_i^* + v_j^*). 
\end{align*}
Where $\theta_{ij} \geq \theta$ is because $e \not\in E_\theta$, and the last equality is due to $\theta = \frac{2}{3}$. 

\subsubsection{Type 1 Paths}\label{pf:thm:mincost:detUB:type1}

Let $P$ be a type 1 path on the vertices $\bigbrace{1,2,...,2\ell}$ where $\bigbrace{(2k+1,2k+2)}_{k = 0}^{\ell-1} \in M$ and $\bigbrace{(2k,2k+1)}_{k=1}^{\ell-1} \in M^*$. Since $P$ contains more edges from $M$ than $M^*$, the set of vertices in $V(P)$ matched by $M^* \cap P$ is a subset of vertices matched by $M \cap P$. Concretely, note that $M$ matches all vertices $\bigbrace{1,2,...,2\ell}$, whereas $M^*$ matches all of the vertices except $\bigbrace{1, 2\ell}$. We thus see that
\begin{align*}
    \sum_{k=1}^{2\ell} v_{k} &= v_1 + v_{2\ell} + \sum_{k=2}^{2\ell - 1} v_k \\
    &\leq v_1^* + v_{2\ell}^* + \sum_{k=2}^{2\ell - 1}  \theta w_k \\
    &\leq v_1^* + v_{2\ell}^* + \sum_{k=2}^{2\ell - 1}  2\theta v_k^* \\
    &\leq 2\theta \sum_{k=1}^{2\ell} v_{k}^*. 
\end{align*}
We conclude by noting that $\theta = \frac{2}{3}$ implies that $2\theta = \frac{4}{3} \leq \frac{3}{2}$. 

\subsubsection{Even length paths: Base case}\label{pf:thm:mincost:detUB:even_base}
We argue by induction on the length of even paths. The base case is when the path $P = \bigbrace{(a,b),(b,c)}$ has 2 edges where $(a,b) \in M, (b,c) \in M^*$. We will use the following fact:

\begin{fact}[Fractions and convex combinations]\label{fact:frac_shift}
If $x > y>0$ and $z > 0$, then $\frac{x+z}{y+z} \leq \frac{x}{y}$.
\end{fact}

\begin{proof}[Proof of Fact \ref{fact:frac_shift}]
This follows from the following algebra:
\begin{align*}
    x > y>0, z > 0 \implies xz &> yz \\
    \implies xy+xz &> xy+yz \\
    \implies x(y+z) &> y(x+z) \\
    \implies \frac{x}{y} &> \frac{x+z}{y+z}.
\end{align*}
\end{proof}

\noindent With this, observe that
\begin{align}\label{eqn:evenpath_base}
    \frac{v_a + v_b + v_c}{v_a^* + v_b^* + v_c^*} &= \frac{w(a,b) + w_c}{w_a + w(b,c)} \leq \frac{w(a,b) + w_c}{w_a + \max(w_b,w_c)}
\end{align}
If the numerator in \eqref{eqn:evenpath_base} is smaller than the denominator, we are done. Therefore from this point onward assume that the numerator is larger. There are two cases. If $w_c \leq w_b$, then we apply this bound to the numerator. If $w_b < w_c$, then we apply Fact \ref{fact:frac_shift} with $z = w_c - w_b$. In either case, we get
\begin{align*}
    \frac{v_a + v_b + v_c}{v_a^* + v_b^* + v_c^*} &\leq \frac{w(a,b) + w_b}{w_a + w_b} \\
    &= \frac{\theta_{ab} w_a + \theta_{ab} w_b + w_b}{w_a + w_b} \\
    &= 1 + \frac{(\theta_{ab}-1) w_a + \theta_{ab}w_b}{w_a+w_b} \\
    &\overset{(a)}{\leq} 1 + \frac{(\theta_{ab}-1) \bigpar{ \frac{1-\theta_{ab}}{\theta_{ab}} } w_b + \theta_{ab}w_b}{\bigpar{ \frac{1-\theta_{ab}}{\theta_{ab}} } w_b+w_b} \\
    &= 1 + \frac{ -(1-\theta_{ab})^2 w_b + \theta_{ab}^2 w_b}{w_b} \\
    &= 1 + \theta_{ab}^2 - \bigpar{1 - 2\theta_{ab} + \theta_{ab}^2 } \\
    &= 2\theta_{ab} \leq 2 \theta = \frac{4}{3} \leq \frac{3}{2}. 
\end{align*}
We get $(a)$ because the numerator is a decreasing function of $w_a$ and the denominator is a increasing function of $w_a$, which means the bound is maximized when $w_a$ is minimized. Furthermore, $\theta_{ab} w_a + \theta_{ab} w_b = w(a,b) \geq w_b$ implies that $w_a \geq \bigpar{ \frac{1-\theta_{ab}}{\theta_{ab}} } w_b$. This establishes the base case for even length paths. 

\subsubsection{Induction Step on the Path Length}\label{pf:thm:mincost:detUB:induction}

Since we already bounded the cost contribution of type 1 paths in Section \ref{pf:thm:mincost:detUB:type1}, in this section we focus on paths that are not type 1. In particular, we focus on paths that have at least as many edges from $M^*$ as they do from $M$. Note that any such path has the following property: Either the first edge in the path or the last edge in the path belong to $M^*$. In this section, we will without loss of generality assume that the first edge of the path belongs to $M^*$. If the first edge belongs to $M$, then it must be the case that the last edge belongs to $M^*$, and we can just reverse the ordering of the edges. \\

\noindent Now we begin the induction step. Assume that for all paths $P$ of length $n$ whose vertices are $\bigbrace{0,1,...,n}$ and whose first edge belongs to $M^*$, we have
\begin{align*}
    w(M;P) &\leq \frac{3}{2} w(M^*;P). 
\end{align*}
Now let $P'$ be a length $n+2$ path whose vertices are $\bigbrace{0,1,2,...,n+1,n+2}$ and whose first edge belongs to $M^*$. Define $P$ to be the path containing the edges $\bigbrace{(i,i+1)}_{i=2}^{n+1}$. Note that the first edge of $P$ belongs to $M^*$, and the length of $P$ is $n$. We can write
\begin{align*}
    w(M;P') = \sum_{i=0}^{n+2} v_i &= v_0 + v_1 + (v_2-w_2) + w_2 + \sum_{i=3}^{n+2} v_i \\
    &= v_0 + v_1 + \underbrace{(v_2-w_2)}_{\text{term 1}} + w(M;P)
\end{align*}
Where term $1$ is because of the following. Vertex $2$ is unmatched in $P \cap M$, so it contributes a cost of $w_2$. However, $(1,2) \in (P' \setminus P) \cap M$ reduces the cost contribution of vertex $2$ from $w_2$ to $v_2$. Similarly, we can write 
\begin{align*}
    w(M^*;P') &= v_0^* + v_1^* + \sum_{i=2}^{n+2} v_i^* \\
    &= v_0^* + v_1^* + w(M^*;P). 
\end{align*}
Since $P$ is a length $n$ path whose first edge belongs to $M^*$, we can use the induction hypothesis to conclude that
\begin{align*}
    w(M;P) \leq \frac{3}{2} w(M^*;P).
\end{align*}
Next, we will show that
\begin{align}\label{eqn:pathlen_induction}
    \frac{v_0 + v_1 + (v_2-w_2)}{v_0^* + v_1^*} \leq \frac{3}{2}.
\end{align}
Once we show \eqref{eqn:pathlen_induction}, we can conclude the proof as follows:
\begin{align*}
    w(M;P') &= v_0 + v_1 + (v_2-w_2) + w(M; P) \\
    &\leq \frac{3}{2} \bigpar{ v_0^* + v_1^* } + \frac{3}{2} w(M^*;P) \\
    &= \frac{3}{2} w(M^*;P')
\end{align*}
Where the inequality is obtained by applying \eqref{eqn:pathlen_induction} and the induction hypothesis on $P$. Therefore, all that remains is to prove \eqref{eqn:pathlen_induction}. Observe that

\begin{align*}
    \frac{v_0 + v_1 + (v_2-w_2)}{v_0^* + v_1^*} = \frac{w_0 + \theta_{12} w_1 + (\theta_{12}-1) w_2}{\theta_{01} w_0 + \theta_{01} w_1} \leq \frac{w_0 + \theta_{12} w_1 + (\theta_{12}-1) w_2}{\max \bigpar{w_0,w_1}}.
\end{align*}
We now use an approach similar to what was done in Section \ref{pf:thm:mincost:detUB:even_base}. If the numerator is smaller than the denominator, then we are done. If it is not, then there are two cases. If $w_0 \leq w_1$, we apply this bound to the numerator. If $w_1 < w_0$, then we apply Fact \ref{fact:frac_shift} with $z=w_0-w_1$. In both cases, we get
\begin{align*}
    \frac{v_0 + v_1 + (v_2-w_2)}{v_0^* + v_1^*} &\leq \frac{w_1 + \theta_{12} w_1 + (\theta_{12}-1) w_2}{w_1} \\
    &\overset{(a)}{\leq} \frac{w_1 + \theta_{12} w_1 + (\theta_{12}-1) \bigpar{ \frac{1-\theta_{12}}{\theta_{12}} } w_1 }{w_1} \\
    &= \frac{\theta_{12} + \theta_{12}^2 - (1-\theta_{12})^2}{\theta_{12}} = \frac{\theta_{12} + \theta_{12}^2 - (1-2\theta_{12} + \theta_{12}^2)}{\theta_{12}} \\
    &= \frac{3\theta_{12}-1}{\theta_{12}} \\
    &\overset{(b)}{\leq} \frac{3\theta-1}{\theta} = \frac{3}{2}.
\end{align*}
Where $(a)$ is because the numerator is a decreasing function of $w_2$, so the bound is maximized when $w_2$ is minimized. Next, $\theta_{12} w_1 + \theta_{12}w_2 = w(1,2) \geq w_1$ implies that $w_2 \geq \bigpar{ \frac{1-\theta_{12}}{\theta_{12}} } w_1$. Finally, $(b)$ is because $\theta_{12} \leq \theta$ since $(1,2) \in M \subset E_\theta$, and the function $\frac{3x-1}{x}$ is increasing on $[\frac{1}{2} ,\frac{2}{3}]$. 

\subsubsection{Base case for type 2 paths of length at least 3}

Our base case is a path $P = \bigbrace{(a,b),(b,c),(c,d)}$ where $(a,b),(c,d) \in M^*$ and $(b,c) \in M$. There are two cases we need to consider:

\begin{enumerate}
    \item At most one of $(a,b)$ or $(c,d)$ is in $E_\theta$.
    \item $(a,b),(c,d) \in E_\theta$.
\end{enumerate}

\begin{remark}
One may be tempted to use the result for type 2 paths of length 1 from Section \ref{pf:thm:mincost:detUB:len1} together with the induction step in Section \ref{pf:thm:mincost:detUB:induction} to establish a result for all type 2 paths. This, however, will not work. In section \ref{pf:thm:mincost:detUB:len1} when discussing type 2 paths of length 1, we used the fact that the edge $e \in M^*$ cannot be in $E_\theta$. Therefore applying the result from Section \ref{pf:thm:mincost:detUB:len1} together with the induction step in Section \ref{pf:thm:mincost:detUB:induction} will only prove the result for type 2 paths that contain at least one edge from $E \setminus E_\theta$. There can in general be type 2 paths whose edges are all contained in $E_\theta$. Case 2 of this section addresses such paths. Once we address case 2, then we can apply the induction argument from Section \ref{pf:thm:mincost:detUB:induction} to establish the desired result for all type 2 paths. 
\end{remark}

\noindent We now prove the desired result for each of the two cases. This section, together with the induction argument from Section \ref{pf:thm:mincost:detUB:induction} establishes the desired result for all type 2 paths. \\

\noindent \textit{Case 1:} In this case, either $(a,b) \not\in E_\theta$ or $(c,d) \not\in E_\theta$. Without loss of generality assume that $(a,b) \not\in E_\theta$ (otherwise simply relabel the vertices). Since $(a,b) \not\in E_\theta$, we have $\frac{w(a,b)}{w_a + w_b} > \theta$. We then observe
\begin{align*}
    v_a + v_b \leq w_a + w_b \leq \frac{1}{\theta} w(i,j) \leq \frac{v_a^* + v_j^*}{\theta} = \frac{3}{2}(v_a^* + v_j^*). 
\end{align*}
We now have a path $P_0 = \bigbrace{(a,b)}$ where $w(M;P_0)$ is no more than $\frac{3}{2} w(M^*;P_0)$. Applying the induction argument from Section \ref{pf:thm:mincost:detUB:induction}, we can conclude that 
\begin{align*}
    \frac{w(M;P)}{w(M^*;P)} &\leq \frac{3}{2}.
\end{align*}
This completes the proof of Case 1. \\

\noindent \textit{Case 2:} In this case we have a path $P = \bigbrace{(a,b),(b,c),(c,d)}$ where $(a,b),(c,d) \in M^* \cap E_\theta$ and $(b,c) \in M$. We first show that the first vertex among $\bigbrace{a,b,c,d}$ to become critical must be either $b$ or $c$. Note that if $a$ were the first vertex among $\bigbrace{a,b,c,d}$ to become critical, then in particular $a$ becomes critical before either $b,c$ becomes critical. Since $(b,c)\in M$, no vertex that becomes critical before $b$ and $c$ choose to match with either of them. Therefore, by the time $a$ becomes critical, both $b,c$ are available, meaning that $a$ should be matched, since in particular it can match to $b$. However, since $a$ is unmatched in $M$, it cannot be the case that $a$ is the first vertex among $\bigbrace{a,b,c,d}$ to become critical. The same argument concludes that $d$ cannot be the first to become critical. This means that the first among these vertices to become critical is either $b$ or $c$. Thus without loss of generality assume that $b$ is the first vertex to become critical. When $b$ becomes critical, both $a,c$ are available. Since $b$ chooses to match with $c$, we can conclude that $\theta_{bc} \leq \theta_{ab}$. \\

\noindent We now observe that
\begin{align*}
    \frac{w(M;P)}{w(M^*;P)} &= \frac{w_a + w(b,c) + w_d}{w(a,b) + w(c,d)} \\
    &= \frac{w_a + \theta_{bc} w_b + \theta_{bc} w_c + w_d}{\theta_{ab} w_a + \theta_{ab} w_b + \theta_{cd} w_c + \theta_{cd} w_d} \\
    &\leq \frac{w_a + \theta_{bc} w_b + \theta_{bc} w_c + w_d}{\theta_{ab} w_a + \theta_{ab} w_b + \max(w_c,w_d)}
\end{align*}
If the numerator is smaller than the denominator, then we are done. So from now on we assume that the numerator is larger. There are two cases to consider. If $w_c > w_d$, then we can bound the numerator by $w_a + \theta_{bc} w_b + \theta_{bc} w_c + w_c$. If $w_d > w_c$, then applying Fact \ref{fact:frac_shift} with $z = w_d - w_c$, we get
\begin{align*}
    \frac{w_a + \theta_{bc} w_b + \theta_{bc} w_c + w_d}{\theta_{ab} w_a + \theta_{ab} w_b + \max(w_c,w_d)} = \frac{w_a + \theta_{bc} w_b + \theta_{bc} w_c + w_c + (w_d-w_c)}{\theta_{ab} w_a + \theta_{ab} w_b + w_c + (w_d - w_c)} \leq \frac{w_a + \theta_{bc} w_b + \theta_{bc} w_c + w_c}{\theta_{ab} w_a + \theta_{ab} w_b + w_c}.
\end{align*}
So in either case, we have 
\begin{align*}
    \frac{w(M;P)}{w(M^*;P)} &\leq \frac{w_a + \theta_{bc} w_b + \theta_{bc} w_c + w_c}{\theta_{ab} w_a + \theta_{ab} w_b + w_c} 
\end{align*}
There are two cases to consider. If $w_a \geq w_c$, then
\begin{align*}
    \frac{w(M;P)}{w(M^*;P)} &\overset{(a)}{\leq} \frac{w_a + \theta_{ab} w_b + \theta_{ab} w_c + w_c}{\theta_{ab} w_a + \theta_{ab} w_b + w_c}\\
    &= 1 + \frac{(1-\theta_{ab}) w_a + \theta_{ab} w_c}{\theta_{ab} w_a + \theta_{ab} w_b + w_c} \\
    &\leq 1 + \frac{(1-\theta_{ab}) w_a + \theta_{ab} w_c}{w_a + w_c} \\
    &= 1 + (1-\theta_{ab}) \frac{w_a}{w_a + w_c} + \theta_{ab} \frac{w_c}{w_a+w_c}\\
    &\leq 1 + \frac{1-\theta_{ab}}{2} + \frac{\theta_{ab}}{2} = \frac{3}{2}.
\end{align*}
Where $(a)$ is due to $\theta_{bc} \leq \theta_{ab}$. The other case is that $w_c > w_a$. In this case, we have
\begin{align*}
    \frac{w(M;P)}{w(M^*;P)} &= \frac{w_a + \theta_{bc} w_b + \theta_{bc} w_c + w_c}{\theta_{ab} w_a + \theta_{ab} w_b + w_c} \overset{(a)}{\leq} \frac{w_a + \theta_{ab} w_b + \theta_{bc} w_c + w_c}{\theta_{ab} w_a + \theta_{ab} w_b + w_c}\\
    &= 1 + \frac{(1-\theta_{ab}) w_a + \theta_{bc} w_c}{\theta_{ab} w_a + \theta_{ab} w_b + w_c} \\
    &= 1 + \underbrace{\frac{(1-\theta_{ab}) w_a}{\theta_{ab} w_a} \frac{\theta_{ab} w_a}{\theta_{ab} w_a + \theta_{ab} w_b + w_c} + \frac{\theta_{bc} w_c}{\theta_{ab} w_b + w_c} \frac{\theta_{ab} w_b + w_c}{\theta_{ab} w_a + \theta_{ab} w_b + w_c}}_{\text{term 1}}
\end{align*}
Where $(a)$ is obtained by bounding $\theta_{bc} w_b \leq \theta_{ab} w_b$. We will show that term 1 is maximized when $w_b$ takes its smallest possible value. To this end, note that term 1 is a convex combination of $\frac{1-\theta_{ab}}{\theta_{ab}}$ and $\frac{\theta_{ab} w_c}{\theta_{ab} w_b + w_c}$. Since $w_b \geq \frac{1-\theta_{bc}}{\theta_{bc}} w_c$ and $\theta_{bc} \leq \theta_{ab}$, we have
\begin{align}\label{eqn:len3base_smallterm}
    \frac{\theta_{bc} w_c}{\theta_{ab} w_b + w_c} \leq \frac{\theta_{bc} w_c}{\theta_{bc} w_b + w_c} \leq \frac{\theta_{bc} w_c}{\theta_{bc} \frac{1-\theta_{bc}}{\theta_{bc}} w_c + w_c} = \frac{\theta_{bc}}{2-\theta_{bc}} \leq \frac{\theta_{ab}}{2-\theta_{ab}}. 
\end{align}
Next, we note that since $\theta_{ab} \leq \theta = \frac{2}{3}$, we have $\frac{\theta_{ab}}{2-\theta_{ab}} \leq \frac{1-\theta_{ab}}{\theta_{ab}}$. So $\frac{1-\theta_{ab}}{\theta_{ab}}$ is larger than $\frac{\theta_{ab} w_c}{\theta_{ab} w_b + w_c}$ for any valid values\footnote{By this we simply mean $\theta_{bc}\leq \theta_{ab} \leq \theta$.} of $w_a,w_b,w_c$. Since $\frac{\theta_{ab} w_a}{\theta_{ab} w_a + \theta_{ab} w_b + w_c}$ and $\frac{\theta_{ab} w_c}{\theta_{ab} w_b + w_c}$ are decreasing functions of $w_b$, we see that for any valid values of $w_a,w_c$, term 1 is maximized when $w_b$ takes its smallest possible value, that is to say $w_b = \max \bigpar{ \frac{1-\theta_{ab}}{\theta_{ab}} w_a, \frac{1-\theta_{bc}}{\theta_{bc}} w_c }$. From this we get:
\begin{align*}
    \frac{w(M;P)}{w(M^*;P)} &=1 + \frac{(1-\theta_{ab}) w_a}{\theta_{ab} w_a} \frac{\theta_{ab} w_a}{\theta_{ab} w_a + \theta_{ab} w_b + w_c} + \frac{\theta_{bc} w_c}{\theta_{ab} w_b + w_c} \frac{\theta_{ab} w_b + w_c}{\theta_{ab} w_a + \theta_{ab} w_b + w_c}\\
    &\overset{(a)}{\leq}1 + \frac{(1-\theta_{ab}) w_a}{\theta_{ab} w_a} \frac{\theta_{ab} w_a}{\theta_{ab} w_a + \theta_{ab} w_b + w_c} + \frac{\theta_{ab}}{2-\theta_{ab}} \frac{\theta_{ab} w_b + w_c}{\theta_{ab} w_a + \theta_{ab} w_b + w_c}\\
    &\overset{(b)}{\leq} 1+\frac{(1-\theta_{ab})}{\theta_{ab}} \frac{\theta_{ab} w_a}{\theta_{ab} w_a + \theta_{ab} \bigpar{ \frac{1-\theta_{ab}}{\theta_{ab}} }w_a + w_c} + \frac{\theta_{ab}}{2-\theta_{ab}} \bigpar{1-\frac{\theta_{ab} w_a}{\theta_{ab} w_a + \theta_{ab} \bigpar{ \frac{1-\theta_{ab}}{\theta_{ab}} }w_a + w_c}} \\
    &= 1+\frac{(1-\theta_{ab})}{\theta_{ab}} \frac{\theta_{ab} w_a}{w_a + w_c} + \frac{\theta_{ab}}{2-\theta_{ab}} \bigpar{1-\frac{\theta_{ab} w_a}{w_a + w_c}} \\
    &\overset{(c)}{\leq} 1+\frac{(1-\theta_{ab})}{\theta_{ab}} \frac{\theta_{ab} w_a}{w_a + w_a} + \frac{\theta_{ab}}{2-\theta_{ab}} \bigpar{ 1 - \frac{\theta_{ab}w_a}{w_a + w_a} } \\
    &= 1 + \frac{1-\theta_{ab}}{2} + \frac{\theta_{ab}}{2-\theta_{ab}} \bigpar{ \frac{2 - \theta_{ab}}{2} } \\
    &= 1 + \frac{1-\theta_{ab}}{2} + \frac{\theta_{ab}}{2} = \frac{3}{2}. 
\end{align*}
We get $(a)$ from the chain of inequalities in \eqref{eqn:len3base_smallterm}. We get $(b)$ and $(c)$ by increasing the weight of $\frac{1-\theta_{ab}}{\theta_{ab}}$ in the convex combination and reduce the weight of $\frac{\theta_{ab}}{2-\theta_{ab}}$ by the same amount. This leads to an upper bound since $\frac{1-\theta_{ab}}{\theta_{ab}} \geq \frac{\theta_{ab}}{2-\theta_{ab}}$. This completes the base case for case 2.



\subsection{Proof of Theorem \ref{thm:mincost:hypergraph:UB}}\label{pf:thm:mincost:hypergraph:UB}

The proof proceeds in the following steps. First, we establish an equivalence between the offline version of the problem \offlinemin{} and the weighted set cover problem through Observation \ref{obs:costmin:k>2:setcov_equiv}. Using this equivalence, we use existing results from set cover to prove that it is NP-hard to approximate \onlinemin{} better than $\log k - O(\log \log k)$. This is done in Lemma \ref{lem:costmin:k>2:polyLB}. We conclude by presenting a randomized batching scheme that can convert any algorithm for weighted set cover into an algorithm for \onlinemin{} with a $2-\frac{1}{d}$ overhead in the competitive ratio. This allows us to establish both the optimal competitive ratio and optimal polynomial-time competitive ratio up to a factor of $2-\frac{1}{d}$. 

We will use \texttt{SC}$(X,\cS)$ to refer to the set cover problem with universe $X$ and a collection of sets $\cS$. We use \texttt{WSC}$(X,\cS,w)$ to denote the weighted version of the problem where $w:\cS \rightarrow \mathbb{R}_+$. We use $k$-\texttt{SC} and $k$-\texttt{WSC} to denote instances where $\abs{S} \leq k$ for all $S \in \cS$. See Appendix~\ref{sec:prelims:setcover} for further details and existing results for set cover.

\begin{observation}[Equivalence between \offlinemin{} and Weighted Set Cover]\label{obs:costmin:k>2:setcov_equiv}
Due to Assumption \ref{assumption:min:split}, \offlinemin{} on an instance $H = (V,E,w)$ is equivalent to $k$-\texttt{WSC}$(V,E,w)$
\end{observation}

\begin{proof}
First we will show that from every matching $M$ for \onlinemin{}, we can construct a cover with the same cost for $k$-\texttt{WSC}$(V,E, w)$. Given any matching $M$ in $H$, let $U$ be the set of vertices that are unmatched under $M$. Defining $\cU := \cup_{i \in U} \bigbrace{i}$, note that $\cE := M \cup \cU$ covers $V$ and 
\begin{align*}
    \sum_{S \in \cE} w(E) &= \sum_{S \in M} w(S) + \sum_{S \in \cU} w(S) = \sum_{e \in M} w(e) + \sum_{i \in U} w(\bigbrace{i}). 
\end{align*}

Next, from every cover $\cE$ of $V$, we  construct a matching in $H$ whose cost is at most the cost of $\cE$ for $k$-\texttt{WSC}$(V,E, w)$. Given a cover $\cE := \bigbrace{S_1,S_2,...,S_m}$, consider the edges $e_i := S_i \setminus \bigpar{ \cup_{j=1}^{i-1} S_j }$ and the set $M := \bigbrace{e_i}_{i=1}^m$. By construction, $\cup_{i=1}^m S_i = \cup_{i=1}^m e_i$, so in particular $M$ is also a cover of $V$. Furthermore, since $e_i \subset S_i$, by Assumption \ref{assumption:min:split}, we have $w(e_i) \leq w(S_i)$, so the cost of $M$ is no greater than the cost of $\cE$. Finally, note that by construction, $e_i \cap e_j = \emptyset$ if $i \neq j$. Therefore $M$ is a matching in $H$ whose cost for \offlinemin{} is no more than the cost of $\cE$ for \texttt{$k$-WSC}. \qed 
\end{proof}

This result allows us to leverage hardness results from set cover \cite{Trevisan01} to prove hardness results for \onlinemin{}. 

\begin{lemma}[Computational Complexity of \onlinemin{}]\label{lem:costmin:k>2:polyLB}
Under Assumption \ref{assumption:min:split}, it is NP-hard to approximate \onlinemin{} within $\log k - O(\log \log k)$. 
\end{lemma}

\begin{proof}
Observation~\ref{obs:costmin:k>2:setcov_equiv} shows an equivalence between \offlinemin{} and $k$-\texttt{WSC}$(X,\cS,w)$ with the additional condition that $\cS$ is downward closed, i.e. $S \in \cS$ implies that $S' \in \cS$ for any subset $S'$ of $S$, and $w(S') \leq w(S)$. This additional condition is due to Assumption \ref{assumption:min:split}. Let $\overline{\cS}$ denote the downward closure of a collection of sets $\cS$. We will show in every instance of $k$-\texttt{SC}, replacing $\cS$ with $\overline{\cS}$ does not change the optimal solution. This proves that all hardness results for $k$-\texttt{SC} also hold for \offlinemin{}. Since \offlinemin{} is a special case of \onlinemin{}, this allows us to inherit hardness results from $k$-\texttt{SC} to \onlinemin{}. 

Now we prove that $k$-\texttt{SC}$(X,\cS)$ and $k$-\texttt{SC}$(X,\overline{\cS})$ have the same optimal solution. Let $\cE$ be any minimum cardinality set cover for $k$-\texttt{SC}$(X,\overline{\cS})$. For any $S \in \overline{\cS} \setminus \cS$, there exists some $T_S \in \cS$ with $S \subset T_S$. Therefore, by replacing every $S \in \cE$ with the corresponding $T_S$, we obtain a set cover for $k$-\texttt{SC}$(X,\cS)$. Thus $k$-\texttt{SC}$(X,\overline{\cS})$ cannot have a strictly smaller optimal solution than $k$-\texttt{SC}$(X,\cS)$. Hence requiring $\cS$ to be downward closed does not make $k$-\texttt{SC} easier. Trevisan proved in \cite{Trevisan01} that it is NP-hard to approximate $k$-\texttt{SC} better than $\log k - O(\log \log k)$. Therefore, by the above reasoning, it is also NP-hard to approximate \onlinemin{} better than $\log k - O (\log \log k)$. \qed 
\end{proof}

We now leverage the relation to set cover to design polynomial time algorithms for \onlinemin{}. In Section \ref{sec:maxcost:algs} we presented \texttt{Randomized-Batching}$(\cdot,\cA)$ for \onlinemax{}, which is described by Algorithm \ref{alg:randombatch}. We will re-purpose this algorithm for \onlinemin{}. Recall that in \texttt{Randomized-Batching}, every $d$ timesteps we use a \offlinemax{} algorithm $\cA$ to compute a matching of the most recently arrived $d$ vertices. We now prove the main result of Theorem \ref{thm:mincost:hypergraph:UB}: \\

\begin{algorithm} 
\caption{\texttt{Weighted-Greedy}$(H)$}\label{alg:setcov_w_greedy}
\textbf{Input:} Shareability Graph $H = (V,E,w)$\;
\textbf{Output:} Matching $M$\;
$M \leftarrow \emptyset$\;
$U \leftarrow V$\;
\For{$e \in E$}{
$\theta_e \leftarrow \frac{w(e)}{\abs{e}}$
}
\While{$U \neq \emptyset$}{
$e^* \leftarrow \arg\min_{e \subset U} \theta_e$\;
$M \leftarrow M \cup e^*$\;
$U \leftarrow U \setminus e^*$
}
\textbf{Return} $M$
\end{algorithm}


\noindent \textbf{Main Result:} Let $\cA$ be an algorithm for $k$-\texttt{WSC} with competitive ratio $\rho(\cA)$. Then \texttt{Randomized-Batching}$(\cdot, \cA)$ is $\bigpar{2 - \frac{1}{d}} \rho(\cA)$-competitive for \onlinemin{}.


\begin{proof}
Let $H = (V,E,w)$ be a shareability graph. For a given value of $z$, recall the definition the batches $\bigbrace{W_{z,i}}_{i=0}^{n/d}$ as $W_{z,i} := \bigpar{ V_{z,i}, E_{z,i}, w }$ where
\begin{align*}
    V_{z,i} &= \bigbrace{ z + i*d + j }_{j=0}^{d-1} \\
    E_{z,i} &= \bigbrace{e \in E : e \subset V_{z,i}}. 
\end{align*}
Let $M^*$ be a minimum weight matching in $H$. Define $\widetilde{M}_z = \cup_{i=1}^{n/d} \widetilde{M}_{z,i}$ where

\begin{align*}
   \widetilde{M}_{z,i} := \bigbrace{ e \cap V_{z,i} : e \in M^*}. 
\end{align*}

\begin{observation}
Note that since $H$ is a shareability hypergraph, for any $e \in E$, $\sum_{i=0}^{n/d} \mathds{1}[e \cap V_{z,i} \neq \emptyset] \leq 2$, i.e. every edge $e \in E$ intersects at most two windows.
\end{observation}

\begin{proof}
We can show this easily by contradiction. Suppose there exists $e\in E$ that intersects more than $2$ windows. Then define
\begin{align*}
    j &= \max \bigbrace{ k : e \cap V_{z,k} \neq \emptyset} \\
    i &= \min \bigbrace{ k : e \cap V_{z,k} \neq \emptyset}.
\end{align*}
It must be the case that $i + 1 < j$. If this is not the case, i.e. if $j=i+1$, then $e$ intersects $V_{u,k}$ only if $k \in \bigbrace{i,i+1}$, and thus it does not intersect more than $2$ windows. Therefore, we know that $i+1 < j$. Finally, note that for any $a \in V_{z,i}, b \in V_{z,j}$, we have $b > a + d$. This implies that $\text{diam}(e) > d$, which is a contradiction to the fact that $\text{diam}(e) < d$ promised by $H$ being a shareability hypergraph. \qed
\end{proof}

Now consider the matching produced by \texttt{Randomized-Batching}$(\cdot, \cA)$ given $Z = z$. This procedure computes a matching $M$ as $\cup_{i=1}^{n/d} M_{z,i}$, where $M_{z,i}$ is a set cover of $W_{z,i}$ chosen by $\cA$. Recall from Observation~\ref{obs:costmin:k>2:setcov_equiv} that under Assumption \ref{assumption:min:split}, any set cover can be converted into a matching with the same cost. We thus see that
\begin{align*}
\mathbb{E} \bigbra{w(M) | Z = z} &= \sum_{i=0}^{n/d} \mathbb{E} \bigbra{w(M_{z,i}) | Z = z} \\
&\overset{(a)}{\leq} \sum_{i=0}^{n/d} \rho(\cA) w(\widetilde{M}_{z,i})  \\
&= \rho(\cA) \sum_{i=0}^{n/d} \sum_{e \in M^*} w(e \cap V_{z,i}) \\
&\leq \rho(\cA) \sum_{e \in M^*} w(e) \sum_{i=0}^{n/d} \mathds{1}[e \cap V_{z,i} \neq \emptyset].
\end{align*}
Where $(a)$ is due to the following. Since $\widetilde{M}_{z,i}$ is a matching in $W_{z,i}$, from Observation~\ref{obs:costmin:k>2:setcov_equiv} we know $w(M_{z,i})$ is at least as large as the minimum weight of a set cover of $(V_{z,i}, E_{z,i},w)$. Since $M_{z,i}$ is the output of a $\rho(\cA)$-competitive $k$-\texttt{WSC} algorithm, we have $w(M_{z,i}) \leq \rho(\cA) w(\widetilde{M}_{z,i})$. We now take expectation over $Z$ to bound $\mathbb{E}[w(M)]$:
\begin{align*}
    \mathbb{E} \bigbra{w(M)} &= \mathbb{E}_Z \bigbra{ \mathbb{E} \bigbra{w(M) | Z = z} } \\
    &\leq \mathbb{E}_Z \bigbra{ \rho(\cA) \sum_{e \in M^*} w(e) \sum_{i=0}^{n/d} \mathds{1}[e \cap V_{Z,i} \neq \emptyset] } \\
    &= \rho(\cA) \sum_{e \in M^*} w(e) \mathbb{E}_Z \bigbra{\sum_{i=0}^{n/d} \mathds{1}[e \cap V_{Z,i} \neq \emptyset] }
\end{align*}
For any edge $e \in E$, let $i$ denote its earliest vertex. Note that when $Z = i \mod d$, $\sum_{i=0}^{n/d} \mathds{1}[e \cap V_{z,i} \neq \emptyset] = 1$. Therefore, $\mathbb{P} \bigpar{\sum_{i=0}^{n/d} \mathds{1}[e \cap V_{z,i} \neq \emptyset] = 1} \geq \frac{1}{d}$. Since we already showed that $\sum_{i=0}^{n/d} \mathds{1}[e \cap V_{z,i} \neq \emptyset]$ is always either $1$ or $2$, this means that the probability it is $2$ is at most $1-\frac{1}{d}$. Applying this gives:
\begin{align*}
    \mathbb{E} \bigbra{w(M)} &\leq \rho(\cA) \sum_{e \in M^*} w(e) \mathbb{E}_Z \bigbra{\sum_{i=0}^{n/d} \mathds{1}[e \cap V_{z,i} \neq \emptyset] } \\
    &\leq \rho(\cA) \sum_{e \in M^*} w(e) \bigbra{ 2 \bigpar{1- \frac{1}{d}} +1 \frac{1}{d} } \\
    &=  \bigpar{2 - \frac{1}{d}} \sum_{e \in M^*} w(e) =  \bigpar{2 - \frac{1}{d}} \rho(\cA) w(M^*).
\end{align*}
This completes the proof. \qed 
\end{proof}


We conclude with several remarks. If $\cA$ computes a minimum weight set cover, i.e. $\rho(\cA) = 1$, then $\texttt{Randomized-Batching}(\cdot, \cA)$ is $(2- \frac{1}{d})$-competitive by Theorem \ref{thm:mincost:hypergraph:UB}. However, it is well known that $k$-\texttt{WSC} is an NP-hard problem, so there is likely no polynomial-time algorithm which can find minimum weight set covers.  \\ 
\indent The \texttt{Weighted-Greedy} algorithm shown in Algorithm \ref{alg:setcov_w_greedy} runs in polynomial-time and is $(1+\log k)$-competitive for $k$-\texttt{WSC}. This is optimal up to lower order terms, given the $\log k - O(\log \log k)$ lower bound of Trevisan \cite{Trevisan01} for $k$-\texttt{SC}, which is a special case of $k$-\texttt{WSC}. The following result is obtained from Theorem \ref{thm:mincost:hypergraph:UB} with $\cA = \texttt{Weighted-Greedy}$. Running Algorithm \ref{alg:randombatch} with \texttt{Weighted-Greedy} as a subroutine gives a $\bigpar{2-\frac{1}{d}}(1 + \log k)$-competitive algorithm for \onlinemin{}. Lemma \ref{lem:costmin:k>2:polyLB} provides a complementary lower bound, establishing the optimal in polynomial-time competitive ratio for \onlinemin{} to be $O(\log k)$. 



\fi 
\end{document}